\documentclass{article}
\usepackage[margin=1in]{geometry}
\usepackage{graphicx} 
\usepackage{amsmath}
\usepackage{amssymb}
\usepackage{amsfonts}
\usepackage{amsthm}
\usepackage{array}
\usepackage{braket}
\usepackage{color}
\usepackage[T1]{fontenc}
\usepackage{lmodern}
\usepackage{microtype}
\usepackage{url}
\usepackage{thmtools}
\usepackage{thm-restate}

\usepackage{mathtools}
\usepackage[hidelinks]{hyperref}
\hypersetup{pdftitle={Planted Cliques and Quantum Symmetry-Adapted Measurements},
  pdfauthor={Vojtech Havlicek, Jordan Docter, Subhash Khot}}

\newtheorem{lemma}{Lemma}
\newtheorem{conjecture}{Conjecture}
\newtheorem*{conjecture*}{Conjecture}
\newtheorem{corollary}{Corollary}
\newtheorem{problem}{Problem}
\newtheorem{proposition}{Proposition}

\providecommand{\ketbra}[2]{\ket{#1}\!\bra{#2}}

\title{Planted Cliques and Quantum Symmetry-Adapted Measurements}

\author{
  Vojtech Havlicek\thanks{IBM Research}\and
  Jordan Docter\thanks{Stanford University} \and
  Subhash Khot\footnotemark[1]
}

\date{September 2026}

\begin{document}

\maketitle

\begin{abstract}

We study how quantum encodings and symmetry-adapted measurements
preserve information for planted-clique detection from one
classical graph. For $k=\lfloor n^{1/2-\varepsilon}\rfloor$,
with fixed $0<\varepsilon<1/2$, detection is statistically
possible but conjectured hard for polynomial-time classical
algorithms. For a compact binary phase encoding, we prove that
constant-advantage detection requires
$\Omega(n^{1+2\varepsilon}\ln^2 n)$ copies of the phase state
of the same graph, even under arbitrary joint measurements.
In the large-copy limit, the optimal decision rule thresholds
the total number of $k$-cliques in the graph and its complement,
but this characterization provides no efficient detector. We therefore explore measurements guided by the symmetries
of the input distributions, starting with the efficient Schur
transform on the full graph register. We show that weak Schur sampling, which measures only the representation label, depends only on edge count and has vanishing distinguishing power in this
regime. When the label and multiplicity registers are discarded, the remaining quantum states are almost perfectly distinguishable. We show that the support of the planted state occupies only a vanishing fraction of the graph Hilbert space. Any subspace containing it still permits near-perfect detection if its relative dimension also vanishes. This gives us freedom to choose a subspace that is easier to measure. We propose exploring subgroup isotypic measurements to find such subspaces. Whether they can yield an efficient detector remains open. Finally, we show that a single supplied coherent quantum sample permits
efficient detection, yielding a conditional computational
separation from one classical sample under quantum
planted-clique hardness.
\end{abstract}

\clearpage
\tableofcontents
\clearpage

\section{Introduction}
\label{sec:introduction}

A central question in quantum computing is which problems can be
solved more efficiently with quantum resources than with classical
computation. Proving such an advantage requires both a quantum
algorithm and a classical lower bound. Since general classical
lower bounds are scarce, many proposed advantages are stated
relative to known algorithms or computational hardness conjectures. Statistical inference provides a concrete setting for this question as
some inference problems exhibit a conjectured
\emph{computational--statistical gap}: a range of parameters in
which accurate inference is statistically possible, but no
polynomial-time algorithm is known~\cite{BPW18,KWB19}.
Motivated by whether quantum computation can make inference efficient
within such a regime, we study the information preserved by quantum
encodings and symmetry-adapted measurements of a planted-clique instance.

We focus on the problem of planted-clique detection
\cite{Jerrum92,Kucera95, Shuichi24, Alon98}.  Under the null hypothesis
$P_0$, the input is a graph drawn from $G(n,1/2)$; an $n$-vertex Erd\H{o}s--R\'enyi graph with
edge probability $1/2$. Under the
alternative $P_1$, we draw the same random graph, choose a
uniformly random set of $k$ vertices, and add all missing edges
within that set. A distinguisher $f$ receives one graph and outputs $1$ to accept
planting. With equal prior probabilities for the two hypotheses,
its advantage is
\begin{align}
    \operatorname{Adv}_n(f)
    &:=\Pr[f(X)=b]-\frac12
    =\frac12\left(
        \Pr_{x\sim P_1}[f(x)=1]
        -\Pr_{x\sim P_0}[f(x)=1]\right),
    \label{eq:detection-advantage}
\end{align}
where $b$ is uniform on $\{0,1\}$, $X\sim P_b$, and the
probabilities include the distinguisher's internal randomness.
Constant advantage means
$\operatorname{Adv}_n(f)\geqslant\eta$ for some fixed $\eta>0$
and all sufficiently large $n$.

Planted clique has a large conjectured computational--statistical
gap. For every fixed $\varepsilon>0$, detection with success
probability $1-o(1)$ is statistically possible when
$k\geqslant(2+\varepsilon)\log_2 n$, because a null graph almost
never contains such a clique. Enumerating all vertex sets of size
$\lceil(2+\varepsilon)\log_2 n\rceil$ and checking whether any
forms a clique gives a quasipolynomial-time distinguisher, with runtime
$n^{O(\log n)}$. Polynomial-time algorithms are known for much larger clique sizes.
For constant-advantage \emph{detection}, simply counting edges suffices
when $k=c\sqrt n$ for any fixed $c>0$~\cite[Appendix B]{Shuichi24}.
The advantage may depend on $c$, but remains bounded away from
zero as $n$ grows. Stronger guarantees are available for
\emph{recovery}, where the goal is to identify the planted clique vertices.
Selecting the vertices of highest degree succeeds with high
probability when $k\geqslant C\sqrt{n\log n}$ for a sufficiently
large constant $C$~\cite{Kucera95}. A basic spectral algorithm
reduces this requirement to $k\geqslant C\sqrt n$~\cite{Alon98},
and message passing achieves recovery with high probability
when $k\geqslant(1+\varepsilon)\sqrt{n/e}$ for any fixed
$\varepsilon>0$~\cite{DM13}. No polynomial-time algorithm for
planted-clique detection is known when
$k=\lfloor n^{1/2-\varepsilon}\rfloor$ for fixed
$0<\varepsilon<1/2$, motivating the following conjecture:    

\begin{conjecture}[Planted Clique Detection Conjecture
{\cite[Conjecture~1]{brennan2019reducibilitycomputationallowerbounds}}]
\label{conj:classical-pc}
For every fixed $0<\varepsilon<1/2$, set
$k=\lfloor n^{1/2-\varepsilon}\rfloor$.
Every uniform randomized polynomial-time algorithm
$f$ satisfies:
\begin{align*}
    \left|
        \Pr_{x\sim P_1}[f(x)=1]
        -\Pr_{x\sim P_0}[f(x)=1]
    \right|&=o(1).
\end{align*}
The probabilities include the randomness of the graph and
of the algorithm. 
\end{conjecture}

Evidence for this conjecture comes from lower bounds for
restricted classical algorithms, including specific
Markov-chain algorithms~\cite{Jerrum92}, sum-of-squares
relaxations~\cite{Barak16}, and statistical-query algorithms
for related distributional formulations of planted
clique~\cite{Feldman12}. Whether comparable barriers persist
for quantum algorithms remains open.
A uniform polynomial-time \emph{quantum} algorithm that achieves
constant-advantage detection at $k=\lfloor n^{1/2-\varepsilon}\rfloor$,
for even one fixed $0<\varepsilon<1/2$, would therefore establish an
average-case quantum advantage under this conjecture.
Strelchuk, Subramanian, and Weso{\l}owski
formulated a quantum planted-clique detection conjecture
\cite[Assumption~2]{SSW26}, extending the hardness assumption
to polynomial-time quantum algorithms. This is an additional
assumption, not known to follow from classical
planted-clique hardness.

Quantum algorithms already provide substantial runtime
improvements for related inference problems.
Hastings~\cite{Hastings20} obtains a quartic speedup for
tensor principal component analysis over the corresponding
classical spectral algorithm, together with an exponential
reduction in space. His analysis also improves the recovery
threshold, but this improvement applies to the classical
algorithms as well.
Schmidhuber, O'Donnell, Kothari, and Babbush~\cite{Schmidhuber24}
obtain a nearly quartic speedup for planted noisy XOR,
also known as sparse learning parity with noise, over the
corresponding classical Kikuchi algorithm at the same
constraint density. Their approach combines a guiding
state constructed from the input with phase estimation
and amplitude amplification.
Schmidhuber and Hastings~\cite{SchmidhuberHastings26}
remove logarithmic losses in the classical Kikuchi hierarchy
for $k$XOR and give a quartic quantum speedup over the
resulting spectral algorithms for slowly growing hierarchy levels.\footnote{These
comparisons are with the corresponding spectral algorithms.
For sufficiently large constant arity, Gupta et al.~\cite{GHOS26}
give a faster classical algorithm that reduces the earlier
quartic quantum speedup for planted noisy XOR to quadratic.}
These quantum speedups improve runtime at the signal
strengths required by the corresponding classical algorithms.
They do not establish polynomial-time quantum inference
in a regime conjectured to require superpolynomial classical
time, and hence do not overcome the conjectured
computational--statistical gap.

Quantum computers can efficiently implement representation-theoretic
transforms, including the Schur transform and the quantum Fourier
transform over the symmetric group~\cite{Beals97,BCH06,burchardt2025highdimensionalquantumschurtransforms}.
These organize quantum states into representation
registers and provide a concrete starting point for designing
measurements that exploit symmetry. Both planted-clique hypotheses
are invariant under vertex relabeling. The null is also invariant
under arbitrary edge-coordinate permutations, whereas the planted
distribution does not have this symmetry. This motivates us to ask whether measurements
built from efficiently implementable representation-theoretic
transforms can expose the planted structure.
Our main contributions are as follows.

\paragraph{Clique qsample access is powerful.} Section~\ref{subsec:clique_qsamples} shows that one
supplied coherent \emph{qsample} permits an $O(n^2)$-time
detector with success probability $1-o(1)$ for clique sizes
$k\geqslant(2+\delta)\log_2n$, with fixed $\delta>0$.
Under quantum planted-clique hardness, this gives a
conditional computational separation between supplied
qsample access and access to one classical graph in
the conjectured hard regime.

\paragraph{Copy complexity of binary phase states.}
Section~\ref{sec:info} encodes edge bits as amplitude signs in a
binary phase state on $O(\log n)$ qubits. For every fixed
$0<\varepsilon<1/2$, at
$k=\lfloor n^{1/2-\varepsilon}\rfloor$, constant-advantage
detection requires
$\Omega(n^{1+2\varepsilon}\ln^2 n)$ copies of the phase state
of the same sampled graph, even under
unrestricted joint measurements
(Proposition~\ref{thm:phase-hidden-copy}).
In the large-copy limit, the optimal rule thresholds the total
number of $k$-cliques in the graph and its complement.
This identifies the statistical optimum but gives no efficient
implementation. We thus turn to measurements guided by
symmetries of the input distributions.

\paragraph{Schur measurements and retained information.}
In Section~\ref{sec:symmetries}, we represent each graph
in the computational basis, using one qubit for each
of the $M=\binom n2$ possible edges.
The null distribution is invariant under arbitrary
edge permutations, whereas the planted distribution
generally is not. We study this action using the Schur
transform, whose output consists of label, Specht, and
multiplicity registers. Weak Schur sampling measures only the representation
label. We show that its outcome distribution depends
only on edge count and has vanishing distinguishing
power when $k=o(\sqrt n)$.
The Specht register, however, retains near-perfect
distinguishability for
$k\geqslant(2+\delta)\log_2n$ with fixed $\delta>0$,
even after both the label and multiplicity registers
are discarded (Corollary~\ref{thm:retain-specht}).

\paragraph{Small planted support and measurement design.}
Section~\ref{sec:level-set-symmetries} shows that, in the
conjectured hard regime, the planted support occupies a vanishing
fraction of the graph Hilbert space, whose dimension is $N=2^M$.
Any subspace containing this support with dimension $o(N)$
permits near-perfect detection against $\sigma_0=I/N$.
This leaves freedom to seek larger accepting subspaces
whose projectors are easier to implement. We propose studying
representation labels for subgroups of the clique-count-preserving
permutation group and give a stabilizer formula for their
probabilities. Finding an informative, efficiently implementable
subgroup measurement remains open.

\paragraph{Related work.}
The distinction between available information and efficient
readout also appears in the hidden subgroup problem.
Ettinger, H{\o}yer, and Knill~\cite{EHK04} give an algorithm
using polynomially many quantum queries in $\log|G|$,
but requiring exponential time.
Childs, Harrow, and Wocjan~\cite{Childs06} show that weak
Fourier--Schur labels can provide insufficient information,
while leaving open useful measurements within the retained
subspaces.
For the hidden-subgroup approach to graph isomorphism,
measurements on individual coset states are insufficient
for efficient detection~\cite{MRS05}; extracting useful
information requires measurements entangled across
$\Omega(n\log n)$ coset states~\cite{HMRRS06}.
Superpolynomial lower bounds also hold for a class of
adaptive quantum sieve algorithms~\cite{Moore_2007}.
We study related limitations for planted-clique detection
from one classical graph, separating the information lost
by particular encodings and readouts from the signal
preserved in the retained states.

\section{Preliminaries}
\label{sec:preliminaries}

\subsection{Planted Clique Detection}

Let $\mathcal E=\binom{[n]}2$ be the possible edge locations and
$M=|\mathcal E|$. Fix an ordering of these coordinates and
represent a graph by $x\in\{0,1\}^M$, with $x_e=1$ when
edge $e$ is present. The null law $P_0=G(n,1/2)$ is the
Erd\H{o}s--R\'enyi distribution: its edge bits are independent
and uniform, so $P_0(x)=2^{-M}$.

For an integer $0\leqslant k\leqslant n$ and a location
$C\in\binom{[n]}k$, let $P_C$ be the law obtained by drawing
from $P_0$ and setting every edge within $C$ to one.
The planted law is $P_1=\mathbb E_C P_C$, where $C$ is
uniform among all $k$-vertex sets.

\begin{problem}[Planted Clique Detection]
\label{problem:decisional_planted_clique}
Given one graph $x\sim P_b$, with $b\in\{0,1\}$ chosen
with equal prior probabilities, decide whether $b=0$ or $b=1$.
\end{problem}
For a hypothesis test $f:\{0,1\}^M\to[0,1]$, let $f(x)$
be its probability of accepting planting on input $x$, and
write $P_b[f]=\mathbb E_{x\sim P_b}f(x)$.
Its \emph{advantage} is its equal-prior success probability
minus $1/2$:
\begin{align*}
    \operatorname{Adv}_n(f)
       &=\Pr[f\text{ correctly identifies }b]-\frac12
        =\frac12\bigl(P_1[f]-P_0[f]\bigr).
\end{align*}
The total-variation distance is
$\operatorname{TV}(P_1,P_0)=\frac12\sum_x|P_1(x)-P_0(x)|$.
The deterministic test $f(x)=\mathbf1[P_1(x)>P_0(x)]$
attains the optimal advantage $\operatorname{TV}(P_1,P_0)/2$.
We seek \emph{constant advantage}: $\operatorname{Adv}_n(f)\geqslant\eta$
for some fixed $\eta>0$ and all sufficiently large $n$.
Conjecture~\ref{conj:classical-pc} rules this out for uniform
randomized polynomial-time algorithms at
$k=\lfloor n^{1/2-\varepsilon}\rfloor$, for each fixed
$0<\varepsilon<1/2$.

\subsection{Planted Clique Detection, Quantumly}
Quantum processing cannot increase the information in a
classical graph, but it may make an informative test easier
to implement. To describe this processing, encode $x$ as
the computational-basis state $\ket{x}$ and average over
the input distribution.
\begin{problem}[Planted Clique Detection, Quantumly]
Let
\begin{align}
    \sigma_b&:=\sum_{x\in\{0,1\}^M}P_b(x)\ket{x}\!\bra{x}.
    \label{Eq:clique_mixture}
\end{align}
Given one copy of $\sigma_b$, with $b\in\{0,1\}$ chosen
with equal prior probabilities, decide whether $b=0$ or $b=1$.
\label{problem:planted_clique_detection_quantumly}
\end{problem}
This is equivalent to receiving one classical graph: measuring
$\sigma_b$ in the computational basis returns $x\sim P_b$,
and a known graph $x$ can be encoded as $\ket{x}$ efficiently.
The algorithm may use ancillas and adaptive measurements;
each result below specifies the restrictions imposed on
its encoding and readout.

A two-outcome measurement $\{I-E,E\}$ accepts planting on
the effect $E$, with advantage
\begin{align*}
    \operatorname{Adv}_n(E)
        &=\frac12\operatorname{Tr}\bigl(E(\sigma_1-\sigma_0)\bigr),
    &0&\leqslant E\leqslant I.
\end{align*}
Writing $D(\sigma_1,\sigma_0)=\frac12\|\sigma_1-\sigma_0\|_1$
for the trace distance, the optimal advantage is
$D(\sigma_1,\sigma_0)/2$. Here the states are diagonal,
so $D(\sigma_1,\sigma_0)=\operatorname{TV}(P_1,P_0)$.

\subsubsection{Quantum Version of the Planted Clique Conjecture}
We use the following
specialization of the quantum hardness assumption of
Strelchuk, Subramanian, and Weso{\l}owski~\cite[Assumption~2]{SSW26}.
\begin{conjecture}[Planted Clique Detection Conjecture, Quantumly]
\label{conj:quantum-pc}
For every fixed $0<\varepsilon<1/2$, set
$k=\lfloor n^{1/2-\varepsilon}\rfloor$.
Every uniform quantum polynomial-time algorithm
$f$ satisfies:
\begin{align*}
    \left|
        P_1[f]
        -P_0[f]
    \right|&=o(1).
\end{align*}
The probabilities include the randomness of the graph and
of the algorithm.
\end{conjecture}

Classical evidence for planted-clique hardness includes
lower bounds for specific Markov-chain algorithms~\cite{Jerrum92},
sum-of-squares relaxations~\cite{Barak16}, and
statistical-query algorithms for related formulations~\cite{Feldman12},
alongside the low-degree framework~\cite{KWB19}.
These results do not automatically extend to quantum algorithms.
To the best of our knowledge, whether these techniques can
establish comparable barriers for restricted quantum algorithms
receiving one classical graph remains open.

\subsection{Worst-case Planted Locations vs Mixtures}
\label{sec:minmax}

Planting a clique at a fixed location $C$ gives the law
$P_C$: all edges within $C$ are present, and the remaining
edges are independent and uniform. Our planted law $P_1$
averages these conditional laws over a uniformly chosen
location. We now show that optimal detection against this
mixture achieves the same distinguishing gap as a single
test required to work for every planted location. Intuitively, it doesn't matter how is the graph labeled. 

\begin{lemma}
\label{lem:classical-minmax}
Taking both suprema over deterministic tests $f$, we have
\begin{align}
\sup_f\inf_C\bigl(P_C[f]-P_0[f]\bigr)
=\sup_f\bigl(P_1[f]-P_0[f]\bigr).
\label{eq:classical-minmax}
\end{align}
\end{lemma}
\begin{proof}
The worst-case gap cannot exceed the average gap, so the
left-hand side is at most the right-hand side.
For the reverse inequality, take the optimal mixture test
$f_\star(x)=\mathbf1[P_1(x)>P_0(x)]$.
Both laws are invariant under vertex relabeling, so
$f_\star$ is invariant as well. Vertex permutations act
transitively on planted locations; consequently,
$P_C[f_\star]$ is the same for every $C$ and equals
$P_1[f_\star]$. Hence
\begin{align*}
    \inf_C\bigl(P_C[f_\star]-P_0[f_\star]\bigr)
        &=P_1[f_\star]-P_0[f_\star]
         =\operatorname{TV}(P_1,P_0).
\end{align*}
This deterministic test attains the optimal mixture gap
against every planted location.
\end{proof}
The same conclusion holds for an encoding compatible with
vertex relabeling. Given pure states $\ket{\psi_x}$, define
\begin{align*}
    \rho_C&:=\mathbb E_{x\sim P_C}\ket{\psi_x}\bra{\psi_x},
    &
    \rho_0&:=\mathbb E_{x\sim P_0}\ket{\psi_x}\bra{\psi_x},
    &
    \rho_1&:=\mathbb E_C\rho_C.
\end{align*}
Assume a unitary representation $V$ of $S_n$ satisfies
$V(\pi)\rho_C V(\pi)^\dagger=\rho_{\pi C}$ and
$V(\pi)\rho_0V(\pi)^\dagger=\rho_0$.
\begin{lemma}
\label{lem:quantum-minmax}
\begin{align}
\sup_{0\leqslant E\leqslant I}\inf_C
 \operatorname{Tr}\bigl[E(\rho_C-\rho_0)\bigr]
=\sup_{0\leqslant E\leqslant I}
 \operatorname{Tr}\bigl[E(\rho_1-\rho_0)\bigr].
\end{align}
\end{lemma}
\begin{proof}
Again, the worst-case gap is at most the average gap.
For the reverse inequality, symmetrize any effect as
\begin{align*}
    E_{\rm sym}
        &:=\mathbb E_{\pi\in S_n}V(\pi)EV(\pi)^\dagger.
\end{align*}
This is still an effect. Covariance and transitivity give
$\operatorname{Tr}(E_{\rm sym}\rho_C)=\operatorname{Tr}(E\rho_1)$
for every $C$, while null invariance gives
$\operatorname{Tr}(E_{\rm sym}\rho_0)=\operatorname{Tr}(E\rho_0)$.
Thus $E_{\rm sym}$ achieves the mixture gap of $E$ against
every location. Taking the supremum proves equality.
\end{proof}

\section{Clique Qsamples are Powerful}
\label{subsec:clique_qsamples}

So far, the input has been one classical graph drawn from
$P_0$ or $P_1$. Supplying a coherent quantum sample
(qsample) changes this access model:
\begin{align}
    \ket{q_b}&:=\sum_x\sqrt{P_b(x)}\ket{x}.
    \label{Eq:qsample}
\end{align}
Measuring $\ket{q_b}$ in the graph basis produces an ordinary sample
from $P_b$, but retaining its coherence permits a different test.
We show that one qsample makes detection efficient whenever a null
graph is unlikely to contain a $k$-clique. Under
Conjecture~\ref{conj:quantum-pc}, this gives a conditional
computational separation.

\begin{proposition}[Detection from one qsample]
\label{prop:clique-qsample-detection}
Fix $\varepsilon>0$ and suppose
$(2+\varepsilon)\log_2n\leqslant k\leqslant n$.
Given one copy of $\ket{q_b}$ from Eq.~\eqref{Eq:qsample},
a uniform quantum algorithm outputs $b$ in $O(n^2)$ time
with equal-prior success probability $1-o(1)$.
\end{proposition}

\begin{proof}
Since $P_0$ is uniform, $\ket{q_0}=\ket{+}^{\otimes M}$.
Apply $H^{\otimes M}$ and accept the null exactly when the
computational-basis outcome is $0^M$; under the null this
occurs with certainty.
Let $A$ be the event that the graph contains a $k$-clique.
Since $P_1(A)=1$, Cauchy--Schwarz gives
\begin{align*}
    \Pr(0^M\mid b=1)
       &=|\braket{q_0 | q_1}|^2
         =\left(\sum_{x\in A}\sqrt{P_0(x)P_1(x)}\right)^2
         \\
       &\leqslant
          \left(\sum_{x\in A}P_0(x)\right)
          \left(\sum_{x\in A}P_1(x)\right)
         =P_0(A)
         \leqslant\binom nk2^{-\binom k2} := \mu_k.
\end{align*}
The final inequality is a union bound over candidate cliques.
Using $\binom nk\leqslant(en/k)^k$, the assumed lower bound on $k$
implies
\begin{align*}
    \mu_k
       &\leqslant
          \left(\frac{en}{k}\,2^{-(k-1)/2}\right)^k \leqslant
          \left(\frac{e\sqrt2}{k}\,n^{-\varepsilon/2}\right)^k
         =o(1).
\end{align*}
The equal-prior success probability is at least $1-\mu_k/2$.
The $M$ Hadamards and the readout take $O(M)=O(n^2)$ time.
\end{proof}
The measurement projects onto the easily prepared null qsample;
it never evaluates whether an individual graph contains a clique.
It requires only one supplied state, with no access to a
state-preparation unitary or its inverse. This runtime
excludes preparation of the supplied qsample. Both access
models permit quantum computation; their difference is
the input resource.

\begin{corollary}[Conditional computational separation at one sample]
\label{cor:clique-qsample-separation}
Fix $0<\varepsilon<1/2$ and
$k=\lfloor n^{1/2-\varepsilon}\rfloor$, and assume Conjecture~\ref{conj:quantum-pc}.
Given one coherent qsample $\ket{q_b}$, a quantum algorithm
distinguishes $P_0$ from $P_1$ in $O(n^2)$ time with success
probability $1-o(1)$.
Given one classical graph $x\sim P_b$, no uniform polynomial-time
quantum algorithm achieves any fixed positive distinguishing
advantage.
\end{corollary}

\paragraph{The computational value of qsample preparation.}
Corollary~\ref{cor:clique-qsample-separation} illustrates the
computational value of a single qsample. Layden et al.~\cite{LSHCN25}
prepare approximate qsamples of suitably regular learned
continuous flow models, while Temme and
Wocjan~\cite{TW25} prepare qsamples of stationary distributions by
quantizing suitable Markov-chain couplings. Our result provides a
concrete conditional example of the value of coherence in such
outputs: one qsample admits an efficient detection measurement,
whereas processing one classical sample remains conjecturally hard
even for a quantum computer. It therefore complements these
preparation results by identifying a task for which the coherent
output can be computationally more useful than a classical draw.
This distinction between statistical information, its quantum
encoding, and its efficient use also connects to work on
quantum-versus-classical learnability~\cite{SG04} and on entangled,
separable, and statistical access to quantum examples~\cite{QSQ23}. 

\section{Binary Phase States Encoding}
\label{sec:info}

Binary phase states encode the edge signs of one graph on
$O(\log n)$ qubits. We study the information available when
the original graph is discarded and detection uses only
copies of this encoding. We first prove a lower bound on
the copies required for constant distinguishing advantage,
allowing arbitrary joint measurements. We then show that,
at fixed $n,k$, infinitely many copies reveal the graph
up to complementation. The resulting optimal decision rule
still depends on finding or counting cliques, so this
characterization does not provide an efficient detector.

\subsection{Binary Phase State Encoding}\label{sec:phase-state-encoding}
As before, let $\mathcal E$ be the set of possible edge locations
in an $n$-vertex graph and $M=\binom n2$. For a graph $x\in\{0,1\}^M$, let
$s_e(x)=2x_e-1$. Its binary phase state is:
\begin{align*}
\ket{\psi_x}&=\frac1{\sqrt M}\sum_{e\in\mathcal E}s_e(x)\ket e,
& \psi_x&=\ket{\psi_x}\bra{\psi_x}.
\end{align*}
The ideal preparation of $t$ copies from $x$ initializes each
edge register in the uniform superposition over $\mathcal E$
and, controlled by the graph register, applies the phase $s_e(x)$
to each basis state $\ket e$. This implements the isometry
\[
    \ket{x}\longmapsto
    \ket{x}\otimes\ket{\psi_x}^{\otimes t}.
\]
The $t$-copy state can be prepared in time polynomial in $M$ and $t$.
The original graph register is then discarded,
and the distinguisher receives only the phase registers.
Its measurement may depend on $n,k,t$, but not on the sampled
graph $x$.
Each phase register uses $O(\log n)$ qubits, and states
encoding different graphs need not be orthogonal. Although
a single copy cannot in general reveal the entire graph,
it may retain information useful for planted-clique detection.
We ask how many copies are required for constant-advantage
detection when arbitrary joint measurements are allowed. Measuring an edge label gives a uniform outcome for every
graph, so detection must exploit relative phases.
Moreover, complementation changes only the global sign:
\begin{align*}
    \ket{\psi_{\bar x}}&=-\ket{\psi_x},
    & \psi_{\bar x}&=\psi_x.
\end{align*}
Thus no number of copies can distinguish $x$ from $\bar x$.
Under hypothesis $b\in\{0,1\}$, the ensemble state is
\begin{align}
\rho_b^{(t)}&:=\mathbb E_{x\sim P_b}[\psi_x^{\otimes t}].
\label{Eq:binary_phase_ensemble}
\end{align}
Detection using this encoding amounts to distinguishing
$\rho_0^{(t)}$ from $\rho_1^{(t)}$.

\subsubsection{Parity Compression}
\label{app:phase-parity}

The amplitude of $\ket{\psi_x}^{\otimes t}$ associated with an ordered edge string $\mathbf e=(e_1,\ldots,e_t) \in {\mathcal E}^t$ is given by:
\begin{align}
    \prod_{j=1}^t \braket{e_j | \psi_x} = \frac{1}{M^{t/2}}\prod_{j=1}^t s_{e_j}(x).
\end{align}
 Since $s_e(x)^2=1$,
repeated occurrences of an edge cancel in pairs and the product
therefore depends only on which edges appear an odd number of
times. Define the parity vector:
\begin{align*}
p(\mathbf e)
&:=\bigoplus_{j=1}^t\mathbf1_{{e_j}}
\in\mathbb F_2^M,
\end{align*}
where $\oplus$ denotes coordinatewise addition modulo two and ${\bf 1}_e$ is a bit vector with $1$ in a coordinate $e$ and zeros otherwise.
The coordinate indexed by an edge is $1$ precisely when that
edge occurs an odd number of times in $\mathbf e$. For example, $(a,a,b)$ and $(c,c,b)$ both have parity vector
$\mathbf1_{{b}}$: the repeated pair cancels.  
With this notation,
\begin{align*}
\prod_{j=1}^t s_{e_j}(x)
&=(-1)^t(-1)^{\langle x,p(\mathbf e)\rangle}.
\end{align*}
Let $w_t(z)$ be the fraction of ordered edge strings in $\mathcal{E}^t$ with a parity
vector $z \in \lbrace 0, 1 \rbrace^M$:
\begin{align}
w_t(z)
&:=\Pr_{\mathbf e\sim\operatorname{Unif}(\mathcal E^t)}
[p(\mathbf e)=z].
\label{Eq:weights}
\end{align}
There are $M^t w_t(z)$ strings in the parity class indexed
by $z$. Edge strings with the same $z$ have equal amplitudes in $\ket{\psi_x}^{\otimes t}$ and we can group these strings into normalized uniform superpositions:
\begin{align*}
    \ket{z;t}
       &:=\frac1{\sqrt{M^t w_t(z)}}
          \sum_{\mathbf e:p(\mathbf e)=z}\ket{\mathbf e},
    & w_t(z)&>0.
\end{align*}
Distinct parity classes have disjoint supports, so the vectors
$\ket{z;t}$ are orthonormal. We identify these vectors with
their parity labels through the isometry
\begin{align}
    \mathcal V_t\ket{z;t}=\ket z,
    \qquad w_t(z)>0,
    \label{Eq:isometry_z}
\end{align}
defined on their span. Every $\ket{\psi_x}^{\otimes t}$ lies
in this span, and $\mathcal V_t$ is independent of $x$.
Thus both $\rho_0^{(t)}$ and $\rho_1^{(t)}$ can be expressed
in the parity basis without changing their trace distance.
We use this isometry to analyze distinguishability;
no efficient implementation is assumed.
Every attainable parity vector satisfies $|z|\leq t$ and
$|z|\equiv t\pmod2$: cancelling repeated edge labels in pairs
leaves exactly $|z|$ occurrences. Define the parity coset
\begin{align}
    H_t&:=\{z\in\mathbb F_2^M:|z|\equiv t\pmod2\},
    & H&:=H_0,
    & |H|&=2^{M-1}.
    \label{eq:parity_coset}
\end{align}
Since $w_t(z)=0$ for unattainable labels, Eq.~\eqref{Eq:isometry_z} gives
\begin{align}
    \mathcal V_t\ket{\psi_x}^{\otimes t}
       &=(-1)^t\sum_{z\in H_t}\sqrt{w_t(z)}
                   (-1)^{\langle x,z\rangle}\ket z.
    \label{eq:phase-compressed-character}
\end{align}
This expression separates the graph-dependent Fourier character
$(-1)^{\langle x,z\rangle}$ from the graph-independent weights
$\sqrt{w_t(z)}$. The distribution $w_t$ is the endpoint
distribution of a random walk on $\mathbb F_2^M$ that starts
at zero and flips one uniformly chosen coordinate at each step.
In Sec.~\ref{sec:phase-optimal-limit}, we use this random-walk
description to characterize the optimal measurement as
$t\to\infty$.

\subsection{Copy Complexity Bounds}
\label{sec:phase-copy-lower}
 Let $m=\binom{k}{2}$ be the number of edges in the planted clique, and write $D^{(t)}:=\frac12\|\rho_1^{(t)}-\rho_0^{(t)}\|_1$.
\subsubsection{Lower Bound}
\label{sec:phase-lower-proof}
\begin{proposition}[Copy lower bound]
\label{thm:phase-hidden-copy}
Assume $3\leqslant k<\sqrt n$. For every integer $t\geqslant1$ such that:
\begin{align}
    t\leqslant
    \frac{M}{16m}\left(\ln\frac{n}{k^2}\right)^2,
    \label{eq:phase-hidden-budget}
\end{align}
the phase-state ensembles satisfy:
\begin{align}
    D^{(t)}\leqslant
    \sqrt{\frac e2}\,\frac{k}{\sqrt n}.
    \label{eq:phase-hidden-distance}
\end{align}
Consequently, for any sequence $k=k(n)\geqslant3$ with
$k=o(\sqrt n)$, constant distinguishing advantage requires
\begin{align*}
    t=\Omega\!\left(\frac{M}{m}\ln^2\frac{n}{k^2}\right).
\end{align*}
In particular, for fixed
$0<\varepsilon<1/2$ and
$k=\lfloor n^{1/2-\varepsilon}\rfloor$, this is
$t=\Omega(n^{1+2\varepsilon}\ln^2 n)$. 
\end{proposition}
\begin{proof}[Proof sketch]
For an edge set $d$, let
$q(d)=\Pr_C[d\subseteq\binom C2]$ be its planted sign
correlation, where $C$ is a uniformly chosen clique location.
All nonempty sign correlations vanish under the null. The main technical step compares the phase ensembles with
an auxiliary product-state ensemble using Poissonization.
A weighted trace-norm inequality then gives
\begin{align*}
    \bigl(D^{(t)}\bigr)^2
        &\leqslant
          \sum_{d\ne\emptyset}\beta^{|d|}q(d)^2,
    & \beta&=\sqrt{1-e^{-8t/M}}.
\end{align*}
Squaring $q(d)$ introduces two independent clique locations
$C,C'$. Their common edges are precisely those within
$C\cap C'$. Writing $R=|C\cap C'|$ and summing over subsets
of these edges yields
\begin{align*}
    \sum_{d\ne\emptyset}\beta^{|d|}q(d)^2
        &=\mathbb E\!\left[
          (1+\beta)^{\binom R2}-1\right].
\end{align*}
Large overlaps are rare. With $a=k^2/n<1$, we have
$\Pr(R=r)\leqslant a^r/r!$. The stated copy budget ensures
$\beta(k-1)\leqslant\ln(1/a)$, so
$(1+\beta)^{\binom r2}\leqslant a^{-r/2}$ for
$2\leqslant r\leqslant k$.
Since overlaps of size zero or one contribute nothing,
\begin{align*}
    \bigl(D^{(t)}\bigr)^2
        &\leqslant\sum_{r=2}^k\frac{a^{r/2}}{r!}
         \leqslant\frac e2\,a.
\end{align*}
Taking square roots gives the claimed bound.
The full proof appears in
Appendix~\ref{app:phase-lower-proof}.
\end{proof}

\subsubsection{Upper Bound}
\label{sec:phase-linear-detection}

\begin{proposition}
\label{prop:phase-linear-upper}
Fix $\epsilon>0$. If $k\geqslant(2+\epsilon)\log_2n$,
then $O(n^2)$ copies of the binary phase state suffice for detection with error
$o(1)$, using separate measurements followed by classical
post-processing.
\end{proposition}

The proof is given in Appendix~\ref{app:phase-linear-upper}.
The construction uses the pair-parity measurement from Hidden
Matching~\cite{BJK08,KR06} followed by a potentially expensive classical search. Together with Proposition~\ref{thm:phase-hidden-copy}, this
establishes copy complexity $\Theta(M)=\Theta(n^2)$ for
$k=\lceil c\log_2n\rceil$ with fixed $c>2$.
For $k=\lfloor n^{1/2-\varepsilon}\rfloor$ with fixed
$0<\varepsilon<1/2$, a gap remains between the lower bound
$\Omega(n^{1+2\varepsilon}\ln^2n)$ and the upper bound
$O(n^2)$.

\subsection{Optimal Measurement in the Large-Copy Limit}
\label{sec:phase-optimal-limit}
We show that for fixed $n,k$, the large-copy limit reveals the sampled graph up to complementation. The optimal equal-prior test then thresholds the total number of $k$-cliques in the
graph and its complement. This characterizes optimal
detection but offers no algorithmic simplification:
the decision rule is itself a clique-count threshold
test, and the analysis supplies no efficient way to
implement it.

\subsubsection{Limiting States}

The weight $w_t(z)$ from Eq.~\eqref{Eq:weights} is the
probability that a walk starting at zero reaches $z$ after
$t$ steps, each of which flips one uniformly chosen edge
coordinate. Since the increments are independent,
$w_t=w_1^{*t}$, where
$w_1(z)=M^{-1}\mathbf1[|z|=1]$ and $*$ denotes convolution
on $\mathbb F_2^M$. The Fourier transform turns convolution into multiplication.
With $\widehat w(y)=\sum_z w(z)(-1)^{\langle y,z\rangle}$:
\begin{align}
    \widehat w_t(y)
        &=\widehat w_1(y)^t
         =\left(\frac1M\sum_{e\in\mathcal E}(-1)^{y_e}\right)^t
         =\left(1-\frac{2|y|}{M}\right)^t.
    \label{eq:walk-fourier}
\end{align}
For $0<|y|<M$, the multiplier has absolute value less than
one, so these modes decay. Only two modes persist:
$\widehat w_t(0)=1$, recording total probability, and
$\widehat w_t(\mathbf1)=(-1)^t$, recording parity. Fourier inversion therefore
gives:
\begin{align*}
    w_t(z)
        &=\frac{1+(-1)^{t+|z|}}{2^M}+o(1).
\end{align*}
The leading term equals $2^{1-M}=|H|^{-1}$ when
$|z|\equiv t\pmod2$ and zero otherwise. Thus the weights
approach the uniform distribution on the parity class
$H_t$. Equivalently, the even and odd subsequences converge
to the uniform laws on their respective parity classes.
See~\cite{LP17} for the general spectral framework.

We first express both parity classes in the same basis,
indexed by $h\in H$. Choose a fixed edge $e_*$ and set
$a_t=0$ for even $t$ and $a_t=\mathbf1_{e_*}$ for odd $t$.
Relabel $\ket{h\oplus a_t}$ as $\ket h$.
This leaves the even-parity class unchanged and shifts
the odd-parity class onto $H$. By Eq.~\eqref{eq:phase-compressed-character}, the relabelled
compressed state is
\begin{align}
    (-1)^{t+\langle x,a_t\rangle}
    \sum_{h\in H}
       \sqrt{w_t(h\oplus a_t)}
       (-1)^{\langle x,h\rangle}\ket h.
       \label{eq:compressed_state_relabelled}
\end{align}
The prefactor is independent of $h$, so it is a global
phase that disappears in the density matrix.
Within the sum, the signs $(-1)^{\langle x,h\rangle}$
are independent of $t$, while
$\sqrt{w_t(h\oplus a_t)}
        \longrightarrow 1/{\sqrt{|H|}}$.
The sum in Eq.~\eqref{eq:compressed_state_relabelled} converges to
\begin{align*}
    \ket{\chi_x}
        &:=\frac1{\sqrt{|H|}}
           \sum_{h\in H}(-1)^{\langle x,h\rangle}\ket h.
\end{align*}
The relabelled
compressed density matrix therefore converges to
$\ket{\chi_x}\bra{\chi_x}$. To determine which graphs these limiting states distinguish,
we compute their overlaps:
\begin{align*}
    \langle\chi_x|\chi_y\rangle
        &=\frac1{|H|}
          \sum_{h\in H}(-1)^{\langle x\oplus y,h\rangle} =\mathbf1[x\oplus y\in\{0,\mathbf1\}].
\end{align*}
If $x\oplus y$ is zero or all ones, every summand equals one.
Otherwise, it has two unequal bits. Flipping the corresponding
coordinates of $h$ preserves even parity and reverses the sign,
so the sum vanishes. Thus different complement pairs give
orthogonal states, while the two graphs within each pair
give the same state. To label each distinct state, and hence each measurement
outcome, exactly once, we choose one graph from every
complement pair. A convenient convention is
\begin{align*}
    \mathcal Y&:=\{y\in\mathbb F_2^M:y_M=0\}.
\end{align*}
Complementation flips the last bit, so each pair has exactly
one representative in $\mathcal Y$. The states
$\{\ket{\chi_y}:y\in\mathcal Y\}$ therefore form an
orthonormal basis of $\mathbb C[H]$. The label $y$ represents the unordered pair
$\{y,y\oplus\mathbf1\}$; it need not be the sampled graph.
Define the probability of this pair under hypothesis $b$ by
\begin{align}
    \nu_b(y)&:=P_b(y)+P_b(y\oplus\mathbf1).
    \label{eq:phase-complement-experiment}
\end{align}
Under planting, a pair can arise from any clique location
contained in either member. To count these possibilities,
define, for any graph $x$,
\begin{align*}
    Z(x)&:=\sum_{C\in\binom{[n]}k}
       \mathbf1\left[
          \text{$x$ is constant on $\binom C2$}
       \right].
\end{align*}
All-one internal bits give a clique in $x$; all-zero bits
give a clique in its complement. Thus $Z(x)$ counts the
$k$-cliques in both graphs and satisfies
$Z(x)=Z(x\oplus\mathbf1)$. For $k\geqslant2$, these are two distinct assignments of
the $m=\binom k2$ internal bits. Under $P_0$, each candidate
set therefore contributes with probability $2^{1-m}$.
Consequently,
\begin{align*}
    \mu&:=\mathbb E_{x\sim P_0}Z(x)
         =\binom nk\,2^{1-m}=2\mu_k.
\end{align*}
Since $P_0$ assigns equal probability to every complement
pair, $\mu$ is also the mean of $Z(y)$ for uniform
$y\in\mathcal Y$. Define the normalized count
$L(y):=Z(y) / \mu$.
The next lemma shows that this is the likelihood ratio
of the planted and null laws on complement pairs.

\begin{lemma}[The limiting experiment]
\label{lem:equidist}
\label{lem:diag-A}
Fix $n,k$ with $2\leqslant k\leqslant n$.
After identifying both parity classes with $H$, the
compressed ensembles converge in trace norm as $t\to\infty$ to
\begin{align}
    \rho_b^{(\infty)}
        &=\sum_{y\in\mathcal Y}\nu_b(y)
          \ket{\chi_y}\bra{\chi_y},
    & b&\in\{0,1\}.
    \label{eq:rho-diag}
\end{align}
The probabilities defined in
Eq.~\eqref{eq:phase-complement-experiment} satisfy
\begin{align*}
    \nu_0(y)&=\frac1{|H|},
    &
    \nu_1(y)&=\frac{L(y)}{|H|}.
\end{align*}
In particular, $L(y)=\nu_1(y)/\nu_0(y)$ and
$\rho_0^{(\infty)}=I/|H|$. The difference
$\Delta_\infty:=\rho_1^{(\infty)}-\rho_0^{(\infty)}$
has eigenvalue $(L(y)-1)/|H|$ on $\ket{\chi_y}$.
\end{lemma}

\begin{proof}
For each graph $x$, the compressed pure-state projector
converges to $\ket{\chi_x}\bra{\chi_x}$.
At fixed $n$, there are finitely many graphs, so averaging
preserves trace-norm convergence. Grouping graphs into
complement pairs gives
\begin{align*}
    \rho_b^{(\infty)}
        &=\sum_x P_b(x)\ket{\chi_x}\bra{\chi_x} =\sum_{y\in\mathcal Y}
          \bigl(P_b(y)+P_b(y\oplus\mathbf1)\bigr)
          \ket{\chi_y}\bra{\chi_y}.
\end{align*}
Under the null, each graph has probability $2^{-M}$,
so each pair has probability $2^{1-M}=|H|^{-1}$. For a fixed planted location $C$, a graph containing the
clique on $C$ has conditional probability $2^{-(M-m)}$;
all other graphs have probability zero.
Each compatible location therefore contributes
$2^{-(M-m)}/\binom nk$ after averaging over $C$.
There are $Z(y)$ such locations across the two members
of the pair, giving
\begin{align*}
    \nu_1(y)
        &=\frac{2^{-(M-m)}}{\binom nk}\,Z(y)
         =\frac{Z(y)}{\mu|H|}
         =\frac{L(y)}{|H|}.
\end{align*}
The formulas for the null state and difference eigenvalues
follow from the common orthonormal basis.
\end{proof}
Both limiting states are diagonal in the character basis,
so measuring this basis preserves their trace distance.
Consequently, optimal detection can be performed in two
steps: identify the complement pair $y$, then decide between
the hypotheses using its probabilities $\nu_0(y)$ and
$\nu_1(y)$.

\subsubsection{Helstrom and Pretty Good Measurements}
\label{sec:pgm}
\label{app:phase-pgm}
\label{sec:trace}

For the limiting states and equal priors, the Helstrom rule
chooses the hypothesis with the larger outcome probability.
Assigning ties to the null gives
\begin{align*}
    p_{\rm H}(y)
        &=\mathbf1[\nu_1(y)>\nu_0(y)]
         =\mathbf1[Z(y)>\mu].
\end{align*}
For the pretty good measurement (PGM), let
$\bar\rho^{(\infty)}
=(\rho_0^{(\infty)}+\rho_1^{(\infty)})/2$.
Its accepting effect is
\begin{align*}
    E_1^{(\infty)}
        &=\tfrac12(\bar\rho^{(\infty)})^{-1/2}
          \rho_1^{(\infty)}(\bar\rho^{(\infty)})^{-1/2}.
\end{align*}
On $\ket{\chi_y}$, the planted and average states have
eigenvalues $\nu_1(y)$ and $(\nu_0(y)+\nu_1(y))/2$.
This gives:
\begin{align*}
    p_{\rm PGM}(y)
        &=\frac{\nu_1(y)}{\nu_0(y)+\nu_1(y)}
         =\frac{Z(y)}{Z(y)+\mu}.
\end{align*}
Thus the PGM accepts with the posterior probability of
planting, while the Helstrom rule chooses the more likely
hypothesis. Both readouts depend on the clique count. When $\mu<1$,
the Helstrom rule simply tests whether $Z>0$, and the
PGM acceptance probabilities differ from it by at most
$\mu/(1+\mu)$.

This analysis identifies the information retained by phase
states, but gives no efficient way to use it for detection:
the limiting decision rules still depend on clique counts.
We next ask whether symmetries of the input distributions
can guide the design of efficient measurements on the full
graph register.

\section{Clique Detection with Schur Sampling}
\label{sec:symmetries}

We return to the mixed graph states $\sigma_b$ from
Section~\ref{sec:preliminaries} and study which
symmetry-adapted readouts preserve the planted signal.
We show that for the efficient Schur transform on the edge qubits,
weak Schur sampling depends only on edge count. This readout has vanishing
distinguishing power when $k=o(\sqrt n)$. For $k\geqslant(2+\varepsilon)\log_2n$, with fixed
$\varepsilon>0$, we show that the Specht register retains
distinguishing distance $1-o(1)$ even after both the label
and multiplicity registers are discarded. The efficient
Schur transform therefore prepares informative reduced
states, which we calculate explicitly. 

\subsection{Distributional Symmetries}
\label{sec:distribution-symmetries}

Both $P_0$ and $P_1$ are invariant under vertex relabeling:
the null is uniform, and the planted clique location is
chosen uniformly. The null is also invariant under arbitrary
permutations of edge coordinates, whereas the planted law
generally is not. Such permutations preserve edge count
but can change clique structure. We ask whether
symmetry-adapted measurements can exploit this broken
symmetry for efficient clique detection.

View $P_0,P_1$ as functions on the graph space
$\Omega=\{0,1\}^M$ and set $N=2^M$. The indicators
$\delta_x(y)=\mathbf1[y=x]$ form the standard basis of
$\mathcal H=\mathbb C[\Omega]$.
An edge permutation $\pi\in S_M$ moves the bit at coordinate
$e$ to coordinate $\pi(e)$, so
$(\pi x)_e=x_{\pi^{-1}(e)}$. The induced operator acts by
\begin{align*}
    R(\pi)\delta_x&=\delta_{\pi x},&
    (R(\pi)f)(x)&=f(\pi^{-1}x).
\end{align*}
Thus $R(\pi)P_b$ is the law obtained by applying $\pi$
to a graph sampled from $P_b$. For $n\geqslant3$, vertex relabeling induces the embedding
$\iota:S_n\hookrightarrow S_M$ given by
\begin{align*}
    \iota(\gamma)\{u,v\}
        &=\{\gamma(u),\gamma(v)\},&
    V(\gamma)&=R(\iota(\gamma)).
\end{align*}
The distributional invariances therefore become
\begin{align*}
    V(\gamma)P_b&=P_b\quad(b=0,1),&
    R(\pi)P_0&=P_0.
\end{align*}

Edge permutations act transitively on graphs with the same
edge count. Averaging over $S_M$ therefore projects
orthogonally onto functions of $|x|$:
\begin{align*}
    \Pi_{\mathrm{edge}}
        &:=\frac1{M!}\sum_{\pi\in S_M}R(\pi), &
    (\Pi_{\mathrm{edge}}P_b)(x)
        &=\frac{\Pr_{P_b}(|X|=|x|)}{\binom M{|x|}}.
\end{align*}
This preserves the edge-count distribution and makes the
graph uniform conditional on its count. Since
$\Pi_{\mathrm{edge}}P_0=P_0$, the planted signal decomposes
orthogonally as
\begin{align}
    P_1-P_0
        &=\bigl(\Pi_{\mathrm{edge}}P_1-P_0\bigr)
          +(I-\Pi_{\mathrm{edge}})P_1.
    \label{eq:symmetry-signal-decomposition}
\end{align}
The first term captures how planting changes the number of edges. The second captures which graphs become more or less likely among graphs with the same number of edges. This structural signal remains invariant under vertex
relabeling. There are also symmetries that permute whole graphs.
Since $P_0$ is uniform and $P_1(x)$ is proportional to
the number of $k$-cliques in $x$, the full common symmetry
group consists of all permutations within clique-count
levels. Section~\ref{sec:level-set-symmetries} considers
this action.

When the probability laws are viewed as vectors, their
Schur decomposition exposes the broken symmetry:
$P_0$ lies entirely in the trivial sector, while the
structural component $(I-\Pi_{\mathrm{edge}})P_1$
lies in nontrivial sectors, where the null has no component.
For mixed graph states, the same symmetries give invariance
under conjugation:
\begin{align*}
    R(\pi)\sigma_0R(\pi)^\dagger&=\sigma_0, &
    V(\gamma)\sigma_bV(\gamma)^\dagger&=\sigma_b.
\end{align*}
Such invariance does not restrict support to invariant
vectors. In particular, $\sigma_0=I/N$ has support in every
sector occurring in $\mathcal H$, so observing a nontrivial
label no longer rules out the null.

We next show that weak Schur sampling on $\sigma_b$ depends
only on the edge-count distribution. It therefore cannot
access the structural signal exposed by the probability-vector
decomposition. We quantify the resulting limitation on
planted-clique detection.

\subsection{Schur Sampling}

The Schur transform simultaneously block-diagonalizes all
edge-permutation operators $R(\pi)$ into irreducible
blocks~\cite{BCH06}. No irreducible block can be split
further by a change of basis that block-diagonalizes
every $R(\pi)$. Grouping equivalent blocks gives
\begin{align*}
    \mathcal H
        &\cong\bigoplus_{i=0}^{\lfloor M/2\rfloor}
          S^{\lambda_i}\otimes\mathcal M_i,
        \qquad \lambda_i=(M-i,i),\\
    R(\pi)
        &\cong\bigoplus_i
          \rho^{\lambda_i}(\pi)\otimes I_{r_i},
        \qquad r_i=M-2i+1.
\end{align*}
Here $S^{\lambda_i}$ is an irreducible representation
of $S_M$, called a \emph{Specht module}, and
$\mathcal M_i\cong\mathbb C^{r_i}$ indexes its $r_i$
equivalent copies. Their tensor product is the
\emph{isotypic component} of type $\lambda_i$. The two factors separate two commuting actions.
Edge permutations act on $S^{\lambda_i}$ and leave
$\mathcal M_i$ unchanged. Collective rotations
$U^{\otimes M}$, with $U\in\mathrm{SU}(2)$, leave
$S^{\lambda_i}$ unchanged and act on $\mathcal M_i$
as the spin-$M/2-i$ representation.

A Schur basis vector is specified by its isotypic
component and a basis index in each factor:
\begin{align*}
    \ket{i,\mu,a}_{\mathrm{Sch}},
    \qquad
    1\leqslant\mu\leqslant r_i,\qquad
    1\leqslant a\leqslant d_i,
    \qquad d_i:=\dim S^{\lambda_i}.
\end{align*}
The label $i$ selects the component; $\mu$ indexes
the multiplicity factor $\mathcal M_i$, and $a$ indexes
the Specht factor $S^{\lambda_i}$. Their ranges depend
on $i$: the tensor-product structure applies within
each component, while the full space is their direct sum.
For circuits, we use a fixed standard encoding of these
indices in registers ordered as
$\mathsf L,\mathsf M,\mathsf S$.
The register sizes accommodate all blocks; basis positions
outside the stated index ranges have zero amplitude. We write
$U_{\mathrm{Sch}}$ for this encoded change of basis,
with initialized ancillas implicit. For qubits, it can
be implemented to accuracy $\varepsilon$ by a circuit
of size $\operatorname{poly}(M,\log(1/\varepsilon))$
\cite{BCH06,BCH07}.

\emph{Weak Schur sampling} measures only the block label $i$.
Equivalently, it projects onto the isotypic components without
resolving the state within them. For $i=0$, the $S_M$
representation is trivial, but its multiplicity is $M+1$.
Its isotypic component is therefore the entire
$(M+1)$-dimensional symmetric subspace
$\operatorname{ran}\Pi_{\mathrm{edge}}$. The null qsample
$\ket+^{\otimes M}$ is one vector in this subspace, so the test from
Section~\ref{subsec:clique_qsamples} resolves \emph{more} than just the weak
label. We now compare the information retained by these
registers on the mixed graph inputs.

\subsubsection{Limitations of Weak Schur Sampling}
\label{sec:gpe}

Let $\Pi_i$ project onto
$S^{\lambda_i}\otimes\mathcal M_i$, and write
\begin{align*}
    p_b(i)&:=\operatorname{Tr}(\Pi_i\sigma_b),&
    D_{\mathsf L}&:=\operatorname{TV}(p_1,p_0).
\end{align*}
We show that this outcome depends on the sampled graph
only through its edge count and cannot reveal
the arrangement of edges or distinguish the hypotheses
better than observing the edge count itself.

\begin{lemma}[Weak Schur sampling depends only on edge count]
\label{prop:edge-covariance}
Suppose the graph encoding $x\mapsto\sigma_x$ satisfies
\begin{align*}
    \sigma_{\pi x}
        &=R(\pi)\sigma_xR(\pi)^\dagger
        \qquad\text{for every }\pi\in S_M,
\end{align*}
and let $\sigma_b=\sum_xP_b(x)\sigma_x$.
Then there is a conditional distribution $p(i\mid w)$
such that
\begin{align*}
    \operatorname{Tr}(\Pi_i\sigma_x)
        &=p(i\mid |x|).
\end{align*}
Consequently, if $W_b$ is the edge count of a graph
sampled under $P_b$, then
$D_{\mathsf L} \leqslant\operatorname{TV}(W_1,W_0)$.
\end{lemma}

\begin{proof}
Each isotypic projector commutes with every edge permutation.
Together with the covariance assumption, this gives
\begin{align*}
    \operatorname{Tr}(\Pi_i\sigma_{\pi x})
        &=\operatorname{Tr}
          \bigl(R(\pi)^\dagger\Pi_iR(\pi)\sigma_x\bigr) =\operatorname{Tr}(\Pi_i\sigma_x).
\end{align*}
Any two graphs with the same edge count are related by
an edge permutation. They therefore have the same label
distribution, which we denote by $p(i\mid w)$. Averaging over the sampled graph and grouping by edge count
gives
\begin{align*}
    p_b(i)
        &=\sum_xP_b(x)p(i\mid |x|) =\sum_{w=0}^{M}\Pr(W_b=w)\,p(i\mid w).
\end{align*}
Thus the label can be generated by first observing the
edge count $w$, then drawing $i$ from $p(\,\cdot\mid w)$.
This second step uses the same rule under both hypotheses.
The triangle inequality and $\sum_i p(i\mid w)=1$ give
\begin{align*}
    D_{\mathsf L}
        &=\frac12\sum_i
          \left|\sum_{w=0}^{M}p(i\mid w)
          \bigl(\Pr(W_1=w)-\Pr(W_0=w)\bigr)\right|\\
        &\leqslant\frac12\sum_{w=0}^{M}
          \left|\Pr(W_1=w)-\Pr(W_0=w)\right|
          \sum_i p(i\mid w) =\operatorname{TV}(W_1,W_0).
\end{align*}
\end{proof}

For the mixed graph states considered here,
$\sigma_x=\ket x\bra x$, the covariance assumption follows
directly from $R(\pi)\ket x=\ket{\pi x}$.
The lemma therefore applies to weak Schur sampling
of $\sigma_0$ and $\sigma_1$. We now show that the edge-count signal vanishes below the
square-root scale. Planting increases the expected edge count
by $m/2=\Theta(k^2)$, while its standard deviation under the
null is $\sqrt M/2=\Theta(n)$. The following bound makes
this comparison precise.

\begin{proposition}[Weak Schur sampling has vanishing advantage]
\label{thm:weak-schur-vanishing}
For $n\geqslant3$ and $2\leqslant k\leqslant n$,
\begin{align*}
    D_{\mathsf L}
        &\leqslant\operatorname{TV}(W_1,W_0)
         \leqslant
         \min\left\{1,\frac{m}{2\sqrt{M-m+1}}\right\}.
\end{align*}
In particular, if $k=o(\sqrt n)$, then
$D_{\mathsf L}=O(k^2/n)=o(1)$.
\end{proposition}

\begin{proof}
Separate the $m$ clique-edge positions from the remaining
$M-m$ positions. Their edge counts satisfy
\begin{align*}
    W_0&\overset{d}=Y+Z,&
    W_1&\overset{d}=Y+m,
\end{align*}
where $Y\sim\operatorname{Bin}(M-m,1/2)$ and
$Z\sim\operatorname{Bin}(m,1/2)$ are independent. A unit shift of a unimodal distribution has total variation
distance equal to its largest probability mass. For $Y$,
\begin{align*}
    \operatorname{TV}(Y,Y+1)
        &=\max_j\Pr(Y=j)
         \leqslant\frac1{\sqrt{M-m+1}}.
\end{align*}
By the triangle inequality, a shift by $r$ has distance
at most $r/\sqrt{M-m+1}$. Averaging over $Z$ therefore gives
\begin{align*}
    \operatorname{TV}(W_1,W_0)
        &\leqslant
          \mathbb E_Z\operatorname{TV}(Y+m,Y+Z) \leqslant
          \frac{\mathbb E[m-Z]}{\sqrt{M-m+1}}
         =\frac{m}{2\sqrt{M-m+1}}.
\end{align*}
Lemma~\ref{prop:edge-covariance} transfers this bound to
the Schur label. When $k=o(\sqrt n)$, we have
$m=O(k^2)$ and $M-m+1=\Theta(n^2)$, proving the claim.
\end{proof}
Similarly to the binary phase state encoding, weak Schur sampling has an additional symmetry: a graph and its complement give the same label distribution.
Appendix~\ref{app:weak-schur-refinement} uses this symmetry
to improve the bound to $D_{\mathsf L}=O(k^4/n^2)$ when $k=o(\sqrt n)$.
Thus, in the conjectured hard regime
$k=\lfloor n^{1/2-\varepsilon}\rfloor$ with fixed
$0<\varepsilon<1/2$, the optimal equal-prior success
probability using only the Schur label is
$\tfrac12+O(n^{-4\varepsilon})$. The exact conditional label probabilities are given in
Appendix~\ref{sec:conditional-label-laws}. 

\subsubsection{Information Retained in the Specht Register}
\label{sec:label-discards}
\label{sec:retained-channels}
\label{sec:specht-discrimination}
\label{sec:retained-signal}

Weak Schur sampling discards the states within each
representation block. We now show that the Specht register
retains near-perfect distinguishability above the logarithmic
clique threshold, even when both the label and multiplicity
are discarded. Using the fixed Schur register encoding, define
\begin{align*}
    \sigma_b^{\mathsf S}
        &:=\operatorname{Tr}_{\mathsf L,\mathsf M}
           \!\left(
              U_{\mathrm{Sch}}\sigma_bU_{\mathrm{Sch}}^\dagger
           \right),
    &
    D_{\mathsf S}
        &:=\tfrac12\|\sigma_1^{\mathsf S}
                       -\sigma_0^{\mathsf S}\|_1.
\end{align*}
The reason the signal survives is geometric: the planted
state has small relative rank, while the discarded registers
have only polynomial dimension. The following lemma makes
this relationship precise.

\begin{lemma}[Discarding a small register]
\label{lem:small-register-rank}
Let $\rho_0=I_N/N$ and let $\rho_1$ be a density operator
of rank $s$ on $\mathbb C^N$. For an isometry
$W:\mathbb C^N\to\mathcal K\otimes\mathcal F$ with
$\dim\mathcal F=d$, set
\[
    \tau_b:=\operatorname{Tr}_{\mathcal F}
            (W\rho_bW^\dagger).
\]
Then
\begin{align*}
    \tfrac12\|\tau_1-\tau_0\|_1
        &\geqslant1-\frac{d^2s}{N}.
\end{align*}
\end{lemma}

\begin{proof}
Let $\Pi$ project onto the support of $\tau_1$.
This test accepts the alternative with certainty;
we bound its acceptance probability under the null. Each pure state in a spectral decomposition of $\rho_1$
has Schmidt rank at most $d$ after applying $W$.
Thus $\operatorname{rank}\tau_1\leqslant sd$.
Also, $WW^\dagger\leqslant I$ gives
\begin{align*}
    \tau_0
        &=\frac1N\operatorname{Tr}_{\mathcal F}(WW^\dagger)
         \leqslant\frac dN I_{\mathcal K}.
\end{align*}
Consequently,
\begin{align*}
    \operatorname{Tr}(\Pi\tau_0)
        &\leqslant\frac dN\operatorname{rank}\tau_1
         \leqslant\frac{d^2s}{N}.
\end{align*}
Since $\operatorname{Tr}(\Pi\tau_1)=1$, the
acceptance-probability gap proves the bound.
\end{proof}

\begin{corollary}[Near-perfect discrimination from the Specht register]
\label{thm:retain-specht}
Let $\mathcal A_k$ be the set of graphs containing a
$k$-clique, and put $c_M:=\sum_i r_i=O(M^2)$. Then
\begin{align*}
    D_{\mathsf S}
        &\geqslant1-c_M^2P_0(\mathcal A_k).
\end{align*}
For every fixed $\varepsilon>0$ and
$k\geqslant(2+\varepsilon)\log_2n$, it follows that
$D_{\mathsf S}=1-o(1)$.
\end{corollary}

\begin{proof}
The discarded label and multiplicity registers have
$c_M=\sum_i r_i$ valid joint basis states $\ket{i,\mu}$.
Unused padded positions have zero amplitude, so the
discarded space can be restricted to dimension $c_M$. The null input is $\sigma_0=I/N$, and the planted input
is supported exactly on graphs in $\mathcal A_k$. Hence
\begin{align*}
    \frac{\operatorname{rank}\sigma_1}{N}
        &=P_0(\mathcal A_k).
\end{align*}
Applying Lemma~\ref{lem:small-register-rank} to the Schur
isometry with $d=c_M$ gives the stated bound. Finally, a union bound over clique locations gives
\begin{align*}
    P_0(\mathcal A_k)
        &\leqslant\binom nk2^{-m}
         =n^{-\Omega_\varepsilon(\log n)}
\end{align*}
throughout the stated regime. This decreases faster than
every inverse polynomial, so multiplication by
$c_M^2=O(M^4)$ still gives $o(1)$.
\end{proof}

\paragraph{Additional analysis.}
Appendix~\ref{app:reduced-schur-states} describes the
retained states explicitly and proves the complementary
comparison: keeping the measured label and multiplicity
is statistically equivalent to observing edge count.
Appendix~\ref{sec:operator-symmetry} explains the distinction
between probability-vector and density-operator symmetries,
and identifies the edge correlations preserved by the
label-and-Specht readout.

\section{Small Planted Support and Measurement Design}
\label{sec:level-set-symmetries}
\label{sec:outcome-measurements}

The preceding proof relies on the small planted support.
Any subspace containing this support permits near-perfect
detection as long as its dimension is $o(N)$.
We therefore have freedom to choose an accepting subspace
whose projector may be easier to implement. 

We give a representation-label test that accepts the
planted support and just one additional direction.
This test achieves near-perfect detection, but we provide
no efficient implementation. We propose subgroup actions
as a possible route to efficient measurements. The open
question is whether their labels can be measured efficiently
and used by an efficient classical decision rule to achieve
constant distinguishing advantage.

\subsection{Small Planted Support}

Assume $3\leqslant k\leqslant n$, and recall that
$M=\binom n2$, $N=2^M$, $m=\binom k2$, and
$\Omega=\{0,1\}^M$.
The number of $k$-cliques in a graph and its null mean are
\begin{align*}
    c_k(x)
        &:=\sum_{C\in\binom{[n]}k}
           \prod_{e\in\binom C2}x_e,
    &
    \mu_k
        &:=\mathbb E_{P_0}[c_k]
         =\binom nk2^{-m}.
\end{align*}
Conditional on planting at $C$, a graph $x$ has probability
$2^{-(M-m)}$ if it contains the clique on $C$, and zero
otherwise. Averaging over the uniformly chosen location gives
\begin{align*}
    P_0(x)&=\frac1N,
    &
    P_1(x)&=\frac{c_k(x)}{N\mu_k},
    &
    \frac{P_1(x)}{P_0(x)}&=\frac{c_k(x)}{\mu_k}.
\end{align*}
To make the support of $\sigma_b$ explicit, group the basis vectors
by clique count. For each attainable value $j$, define
\begin{align*}
    \mathcal X_j&:=\{x:c_k(x)=j\},
    &
    E_j&:=\sum_{x\in\mathcal X_j}\ketbra{x}{x}.
\end{align*}
The projector $E_j$ has rank $N_j:=|\mathcal X_j|$.
Since each law is constant within a level, the density
matrices decompose as
\begin{align*}
    \sigma_0
        &=\frac1N\sum_j E_j=\frac IN,
    &
    \sigma_1
        &=\frac1{N\mu_k}\sum_j jE_j.
\end{align*}
Thus the null has full support on every level but the planted state vanishes on the clique-free level
$j=0$ and has full support within every positive level:
\begin{align*}
    \ker\sigma_1&=\operatorname{ran}E_0,
    &
    \operatorname{supp}\sigma_1
        &=\bigoplus_{j>0}\operatorname{ran}E_j.
\end{align*}
The following lemma
bounds the fraction of the Hilbert space occupied
by this support.

\begin{lemma}[Small relative rank of the planted state]
\label{lem:outcome-rank-ratio}
Let $\mathcal A_k:=\{x:c_k(x)>0\}$ be the graphs containing
a $k$-clique. Then
\begin{align*}
    \operatorname{rank}\sigma_0
        &=2^{\binom n2},
    &
    \operatorname{rank}\sigma_1
        &=2^{\binom n2}P_0(\mathcal A_k).
\end{align*}
The fraction of the graph Hilbert space occupied by the
planted support satisfies
\begin{align*}
    2^{-\binom k2}
        &\leqslant
          \frac{\operatorname{rank}\sigma_1}
               {\operatorname{rank}\sigma_0}
         =P_0(\mathcal A_k)
         \leqslant\binom nk\,2^{-\binom k2}.
\end{align*}
In particular, for
$k=\lfloor n^{1/2-\varepsilon}\rfloor$ with fixed
$0<\varepsilon<1/2$,
\begin{align*}
    \frac{\operatorname{rank}\sigma_1}
         {\operatorname{rank}\sigma_0}
        &=2^{-(\frac12+o(1))n^{1-2\varepsilon}}.
\end{align*}
\end{lemma}

\begin{proof}
The diagonal entries of $\sigma_1$ are positive exactly
on $\mathcal A_k$, proving the rank identities.
Under $P_0$, any fixed candidate set forms a clique with
probability $2^{-m}$. This gives the lower bound;
a union bound over all $\binom nk$ candidates gives
the upper bound $\mu_k$.
At the stated clique size,
\begin{align*}
    m = \binom k2
        &=\left(\frac12+o(1)\right)n^{1-2\varepsilon},
    &
    \log_2\binom nk
        &=O\!\left(n^{1/2-\varepsilon}\log n\right)
         =o\!\left(n^{1-2\varepsilon}\right).
\end{align*}
Thus the lower and upper bounds have the same leading
exponent, proving the asymptotic formula.
\end{proof}

Whenever $\mu_k=o(1)$, the planted support therefore occupies
a vanishing fraction of the Hilbert space. This includes
$k\geqslant(2+\varepsilon)\log_2n$ for every fixed
$\varepsilon>0$.
To see the freedom this gives a measurement, let $\mathcal V$
be any subspace containing $\operatorname{supp}\sigma_1$,
and let $\Pi_{\mathcal V}$ be its projector. Accepting
planting on this subspace gives
\begin{align*}
    \operatorname{Tr}(\Pi_{\mathcal V}\sigma_1)&=1,
    &
    \operatorname{Tr}(\Pi_{\mathcal V}\sigma_0)
        &=\frac{\dim\mathcal V}{N}.
\end{align*}
The distinguishing gap is therefore $1-\dim\mathcal V/N$.
Any such subspace with $\dim\mathcal V=o(N)$ still permits
near-perfect detection; dimension at most $(1-\eta)N$
suffices for gap $\eta$, for any fixed $0<\eta<1$.

\paragraph{An isotypic test.} Let
$G=\prod_{j:\,\mathcal X_j\ne\emptyset}S_{\mathcal X_j}$
independently permute graphs within each clique-count level.
The $k$-clique-free subspace decomposes as
\begin{align*}
    \ker\sigma_1
        &=\operatorname{span}\{\ket{u_0}\}\oplus W_0,
\end{align*}
where $\ket{u_0}$ is its normalized uniform vector.
The space $W_0$ carries the standard irreducible
representation of $S_{\mathcal X_0}$~\cite{Sagan01}.
Since $k\geqslant3$, the empty graph and all one-edge
graphs are clique-free, so $N_0\geqslant M+1$ and this
representation is nontrivial. No copy occurs in
$W_0^\perp$, where $S_{\mathcal X_0}$ acts trivially.
Thus $W_0$ is a full isotypic component of $G$.
Rejecting planting on this label and accepting all others
gives a projector $F_G$ onto
$\operatorname{supp}\sigma_1
 \oplus\operatorname{span}\{\ket{u_0}\}$.
Its accepting subspace contains the entire planted support
and just one clique-free direction. Consequently,
\begin{align*}
    \operatorname{Tr}(F_G\sigma_1)&=1,
    &
    \operatorname{Tr}(F_G\sigma_0)
        &=\frac{\operatorname{rank}\sigma_1+1}{N}
         \leqslant\mu_k+\frac1N.
\end{align*}
Thus this label test permits near-perfect detection
when $\mu_k=o(1)$. However, a direct generalized phase-estimation construction using
a register indexed by all elements of $G$ would require
at least:
\begin{align*}
    \log_2|G|
        &\geqslant\log_2(N_0!)
         =\Omega(M2^M)
\end{align*}
qubits when $\mu_k=o(1)$, making this construction inefficient. There is also a worst-case obstacle to implementing $F_G$.
On a computational-basis graph input,
\begin{align*}
    \bra{x}F_G\ket{x}
        &=
        \begin{cases}
            1, & c_k(x)>0,\\
            1/N_0, & c_k(x)=0.
        \end{cases}
\end{align*}
Since $N_0\geqslant M+1\geqslant4$ for $k\geqslant3$,
these probabilities distinguish clique existence with
bounded error. A uniform polynomial-size implementation
for all $n,k$, even allowing sufficiently small constant
additive error in the acceptance probabilities, would
therefore imply $\mathrm{NP}\subseteq\mathrm{BQP}$~\cite{Karp72}.

Detection requires only a constant gap between average
acceptance probabilities under $P_1$ and $P_0$ and 
it need not reproduce $F_G$ on every graph.
We therefore turn to subgroup actions whose labels might
be both informative and efficiently measurable.

\subsection{Subgroup Actions: a Possible Path to Efficient Measurement}
\label{sec:outcome-implementation-target}
\label{sec:outcome-subgroup-criterion}

A possible route to efficient detection is to search for
subgroups $H\leqslant G$ whose representation labels are
informative and can be measured and processed efficiently.
The small planted support leaves flexibility in this search:
an accepting effect $0\leqslant F\leqslant I$ need only satisfy
\begin{align*}
    \operatorname{Tr}\bigl(F(\sigma_1-\sigma_0)\bigr)
        &=\operatorname{Tr}(F\sigma_1)
          -\frac{\operatorname{Tr}F}{N}
         \geqslant\eta,
\end{align*}
for some constant $\eta>0$ independent of $n$.
For a projector, this asks for a subspace whose planted
probability exceeds its relative dimension by a constant.
It need not contain the entire planted support. A first question is which subgroups retain useful label
information. Restricting to a subgroup can place the same
irreducible type in both the planted support and its kernel,
so preserving clique count alone does not suffice.
The following formula provides a way to assess candidate
actions before attempting to construct their measurements.

For a graph $x$, let $H_x:=\{h\in H:hx=x\}$ be its
stabilizer and $H\cdot x$ its orbit.
For an irreducible $H$-representation $V_\lambda$ of
dimension $d_\lambda$, let $V_\lambda^{H_x}$ denote
the subspace fixed by every element of $H_x$.

\begin{lemma}[Label probabilities from stabilizers]
\label{lem:outcome-stabilizers}
For the action $U(h)\ket{x}=\ket{hx}$, the isotypic
measurement on $\ket{x}$ returns $\lambda$ with probability
\begin{align*}
    p^H(\lambda\mid x)
        &=\frac{d_\lambda}{|H\cdot x|}
          \dim V_\lambda^{H_x}.
\end{align*}
\end{lemma}

\begin{proof}
Let $\chi_\lambda$ be the character of $V_\lambda$.
The character formula for the isotypic projector gives
\begin{align*}
    p^H(\lambda\mid x)
        &=\frac{d_\lambda}{|H|}
          \sum_{h\in H_x}\overline{\chi_\lambda(h)}
         =\frac{d_\lambda|H_x|}{|H|}
          \dim V_\lambda^{H_x}.
\end{align*}
Only stabilizer elements contribute, since
$\bra{x}U(h)\ket{x}=\mathbf1[hx=x]$.
The character average gives the dimension of the fixed
subspace, and the orbit--stabilizer identity completes
the proof.
\end{proof}

See also~\cite[Eq.~(17)]{LH25} for the equivalent
induction-sampling formula.
For a candidate subgroup, the ensemble label laws are $p_b^H(\lambda) =\mathbb E_{X\sim P_b}[p^H(\lambda\mid X)]$.
Their total variation distance is the largest distinguishing
gap obtainable from the label. Thus the statistical target
is $\operatorname{TV}(p_1^H,p_0^H)=\Omega(1)$. The formula also gives a useful obstruction for this search.
Every graph with trivial stabilizer has the same label law:
\begin{align*}
    H_x=\{1\}
    \quad\Longrightarrow\quad
    p^H(\lambda\mid x)=\frac{d_\lambda^2}{|H|}.
\end{align*}
For vertex relabeling, $H_x=\operatorname{Aut}(x)$,
so all asymmetric graphs give identical label statistics.
Consequently,
\begin{align*}
    \operatorname{TV}(p_1^H,p_0^H)
        &\leqslant
          P_0\bigl(\operatorname{Aut}(X)\ne\{1\}\bigr)
          +P_1\bigl(\operatorname{Aut}(X)\ne\{1\}\bigr).
\end{align*}
If asymmetric graphs dominate both ensembles, this action
cannot provide informative labels. The stabilizer formula helps assess whether a proposed
subgroup action gives informative labels. To obtain an
efficient detector, we must also be able to measure these
labels and decide from the outcome in polynomial time.
We therefore ask whether a subgroup can satisfy all three
requirements in the conjectured hard regime:

\begin{problem}[Efficient clique detection from subgroup labels]
\label{problem:efficient-subgroup-detection}
Fix $0<\varepsilon<1/2$ and
$k=\lfloor n^{1/2-\varepsilon}\rfloor$.
Can one find subgroups $H_n\leqslant G$, where $G$ permutes
graphs within their $k$-clique-count levels, with the
following properties?
\begin{enumerate}
    \item The $H_n$-isotypic label can be measured by a uniform
    polynomial-time quantum algorithm, up to $o(1)$
    total variation error on each ensemble $\sigma_b$.

    \item A uniform polynomial-time classical test,
    accepting with probability $f_n(\lambda)$ upon
    observing label $\lambda$, achieves
    \begin{align*}
        \sum_\lambda f_n(\lambda)
        \bigl(p_1^{H_n}(\lambda)-p_0^{H_n}(\lambda)\bigr)
            &\geqslant\eta
    \end{align*}
    for some constant $\eta>0$ and all sufficiently large $n$.
\end{enumerate}
Here
$p_b^{H_n}(\lambda)
=\operatorname{Tr}(\Pi_\lambda^{H_n}\sigma_b)$
is the ideal label distribution.
\end{problem}

\section{Conclusion}
\label{sec:conclusion}

We studied how quantum encodings and symmetry-adapted
measurements preserve information for planted-clique
detection from one classical graph. Our motivating question
is whether quantum computation can overcome the conjectured
computational--statistical gap at
$k=\lfloor n^{1/2-\varepsilon}\rfloor$, with fixed
$0<\varepsilon<1/2$. In this regime, detection is statistically
possible with near certainty but conjectured to be
classically intractable.

For the binary phase encoding, constant-advantage detection
requires $\Omega(n^{1+2\varepsilon}\ln^2 n)$ copies of the
state of the same sampled graph, even under arbitrary joint
measurements.
In the large-copy limit, the encoding reveals the graph
up to complementation. The optimal decision rule thresholds
the total number of $k$-cliques in the graph and its complement;
identifying this rule does not provide an efficient detector.

On the full graph register, weak Schur sampling depends
only on edge count and has vanishing distinguishing power
in the conjectured hard regime. Retaining the measured label
and Specht register while discarding multiplicity preserves
near-perfect distinguishability. This remains true even
after discarding the label. The efficient Schur transform
therefore prepares informative reduced states, and our
explicit formulas identify concrete measurement targets.
The remaining challenge is to extract their signal efficiently.

The small rank of the planted state gives a geometric
perspective on measurement design. While the null is
maximally mixed, $\sigma_0=I/N$, the planted state has
rank $o(N)$. Any subspace containing its support with
dimension $o(N)$ still permits near-perfect detection.
This leaves freedom to seek larger accepting subspaces
whose projectors are easier to implement.
The full group of clique-count-preserving permutations
provides such a test through its representation label.
We propose seeking efficient measurements through subgroup
actions, using our stabilizer formula to assess their
label statistics. An algorithm additionally requires an
efficient label measurement and an efficient decision rule.

Finally, one supplied coherent qsample permits efficient
detection. Under quantum planted-clique hardness, this
gives a conditional computational separation from access
to one classical graph, emphasizing the importance of
the input model. Efficient quantum detection from one
classical graph in the conjectured hard regime remains open.

\subsection{Open directions}

\paragraph{Efficient subgroup isotypic sampling.}
The small support of the planted ensemble $\sigma_1$ identified in
Section~\ref{sec:level-set-symmetries} leaves freedom
to seek larger accepting subspaces that are easier to
measure. Problem~\ref{problem:efficient-subgroup-detection}
proposes exploring this freedom through isotypic measurements
for subgroups of the clique-count-preserving group $G$.
We seek a family $H_n\leqslant G$ whose representation
label admits both an efficient measurement and an efficient
classical acceptance rule $f_n(\lambda)\in[0,1]$ satisfying
\begin{align*}
    \sum_\lambda f_n(\lambda)
       \bigl(p_1^{H_n}(\lambda)-p_0^{H_n}(\lambda)\bigr)
        &\geqslant\eta,
\end{align*}
where $\eta>0$ is independent of $n$ and $p_b^{H_n}$ is
the label distribution under $\sigma_b$.
Both procedures must be uniform and polynomial-time.
The stabilizer formula provides a way to assess whether
a proposed action has the required statistical signal.

\paragraph{Lower bounds for restricted quantum algorithms.}
Our phase-copy lower bound and the vanishing advantage
of weak Schur sampling establish limitations of two
natural approaches to quantum detection. Although the
retained Schur states are nearly perfectly distinguishable,
an efficient measurement accessing their signal may not
exist. These findings are consistent with the possibility
that the computational--statistical gap persists under
quantum computation, as conjectured by Strelchuk,
Subramanian, and Weso{\l}owski~\cite{SSW26}.

Can the classical techniques supporting planted-clique
hardness yield lower bounds for restricted classes of
quantum algorithms? Potential starting points include
bounds for specific Markov-chain algorithms~\cite{Jerrum92},
sum-of-squares relaxations~\cite{Barak16}, and
statistical-query algorithms for related planted-clique
formulations~\cite{Feldman12}, as well as the low-degree
framework~\cite{KWB19}.
Such extensions could provide evidence for quantum
hardness beyond the particular encodings and readouts
studied here.

\subsection{Statement on AI usage}
We used Claude 5.1 Fable and ChatGPT 5.6 Sol for assistance with
exposition, proofs, argument checking, and editing. The authors
are responsible for the results and their verification.
The idea to study planted clique as a target for quantum
algorithms is ours. Most proofs in
Section~\ref{sec:info} were developed without generative AI,
except that GPT-5.6 suggested the Poissonization step in the
phase-copy lower bound that we liked. The ideas in Section~\ref{sec:symmetries}
were also primarily ours; GPT-5.6 helped organize their exposition and fill in details of many of the intuitive arguments.
We then made extensive use of ChatGPT-6 Astra at Max or Ultra settings
during the writeup, with author guidance throughout. Making AI edited text and proofs human-readable is still more work than we thought. Stressed from the recent arXiv delays and not meeting a conference submission deadline, VH uploaded two fairly unpolished versions of this work way too early. It was a mistake and something he will avoid in the future. This is a v3 and it may not be the last one.

\section*{Acknowledgments}
We thank Tselil Schramm and Sergey Bravyi for helpful discussions. 
This material is based upon work partially supported by the National Science Foundation under award No. 2016245 and 2440805, by the Air Force Office of Scientific Research under grant agreement FA9550-21-1-0392, and by the Shoucheng Zhang Graduate Fellowship.

\clearpage
\appendix

\section{Binary Phase States: Proofs}
\label{sec:reconstruction}

\subsection{Copy Lower Bound}
\label{app:phase-lower-proof}
We prove Proposition~\ref{thm:phase-hidden-copy} in three steps:
compute the planted sign correlations, bound the trace distance
by a weighted sum of their squares, and estimate this sum through
clique overlaps. Throughout, all copies encode the same graph $x$,
and $s_e(x)=2x_e-1$ as in Section~\ref{sec:phase-state-encoding}.

\paragraph{Step 1: Sign correlations.}
\label{sec:paritybasis}
For an edge set $d\subseteq\mathcal E$, the null expectation of
$\prod_{e\in d}s_e(x)$ is zero unless $d=\emptyset$.
Conditional on planting at $C$, the signs inside $\binom C2$
are all $+1$, while the remaining signs are independent and fair.
Thus
\begin{align*}
    q(d)&:=\mathbb E_{P_1}\prod_{e\in d}s_e(x)
         =\Pr_C\!\left[d\subseteq\binom C2\right]
         =\frac{\binom{n-r(d)}{k-r(d)}}{\binom nk},
    & r(d)&:=\left|\bigcup_{e\in d}e\right|.
\end{align*}
Here $q(\emptyset)=1$, and $q(d)=0$ if $r(d)>k$.
The probability depends on the number of vertices touched by $d$:
two adjacent edges require three specified vertices in $C$,
whereas two disjoint edges require four.

In the compressed parity basis of
Eq.~\eqref{eq:phase-compressed-character}, the difference
$d=z\oplus z'$ has even weight. Hence
$(-1)^{\langle x,d\rangle}=\prod_{e\in d}s_e(x)$, and averaging
the pure-state matrix entries gives
\begin{align*}
    (\rho_0^{(t)})_{zz'}
        &=w_t(z)\mathbf1[z=z'], &
    (\rho_1^{(t)})_{zz'}
        &=\sqrt{w_t(z)w_t(z')}\,q(z\oplus z').
\end{align*}
Both states have the same diagonal. Their difference lies
in the off-diagonal correlations.

\paragraph{Step 2: A weighted correlation bound.}
We first record a trace-norm inequality, then apply it to an
auxiliary ensemble with weights that factor over edges.

\begin{lemma}[Weighted trace-norm inequality]
\label{lem:phase-weighted-trace-norm}
Let $W=\operatorname{diag}(w)$ be a positive definite density
matrix. For every matrix $X$ on the same space,
\begin{align*}
    \|X\|_1^2
        &\leqslant\|W^{-1/4}XW^{-1/4}\|_2^2
         =\sum_{z,z'}\frac{|X_{zz'}|^2}{\sqrt{w(z)w(z')}}.
\end{align*}
\end{lemma}
\begin{proof}
Set $Y=W^{-1/4}XW^{-1/4}$. Schatten H\"older's inequality gives
\begin{align*}
    \|X\|_1
        &=\|W^{1/4}YW^{1/4}\|_1
         \leqslant\|W^{1/4}\|_4^2\|Y\|_2
         =\|Y\|_2,
\end{align*}
since $\operatorname{Tr}W=1$.
The entrywise formula follows by expanding $\|Y\|_2^2$.
\end{proof}
For $X=\rho-W$, this is also a special case of the quantum
$\chi^2$ bound in~\cite[Lemma~5 and Eq.~(7)]{Temme10}.

\begin{lemma}[Weighted Fourier bound]
\label{lem:phase-overlap-bound}
For every integer $t\geqslant1$, set
$\beta=\sqrt{1-e^{-8t/M}}$. Then
\begin{align*}
    \bigl(D^{(t)}\bigr)^2
        &\leqslant\sum_{\substack{d\subseteq\mathcal E\\d\ne\emptyset}}
          \beta^{|d|}q(d)^2.
\end{align*}
\end{lemma}
\begin{proof}
\emph{Comparison with a product-state ensemble.}
Set
\begin{align*}
    p&:=\frac{1-e^{-4t/M}}2, &
    \ket{\varphi_x}
        &:=\bigotimes_{e\in\mathcal E}
          \left(\sqrt{1-p}\ket0+s_e(x)\sqrt p\ket1\right).
\end{align*}
We construct a channel, independent of $x$, that converts this
state into $t$ phase copies with probability greater than $1/2$.
Consider the normalized state
\begin{align*}
    \ket{\Phi_x}
        &:=\sum_{j=0}^{\infty}
          \sqrt{e^{-2t}\frac{(2t)^j}{j!}}\,
          \ket j\ket{\psi_x}^{\otimes j},
\end{align*}
whose orthogonal sectors have copy number
$T\sim\operatorname{Poi}(2t)$; the $j=0$ sector is a fixed vacuum.
If $x,x'$ differ on $h$ coordinates, then
\begin{align*}
    \langle\varphi_x|\varphi_{x'}\rangle
        &=(1-2p)^h=e^{-4th/M},\\
    \langle\Phi_x|\Phi_{x'}\rangle
        &=e^{-2t}\sum_{j=0}^{\infty}
          \frac{(2t)^j}{j!}(1-2h/M)^j=e^{-4th/M}.
\end{align*}
The two families have the same Gram matrix, so an isometry
independent of $x$ maps $\ket{\varphi_x}$ to $\ket{\Phi_x}$.
Measure $T$ after this isometry. If $T\geqslant t$, keep $t$
copies; otherwise output a fixed failure state. The success
probability $\gamma:=\Pr(T\geqslant t)$ is independent of $x$, and
\begin{align*}
    1-\gamma
        &=\Pr(T<t)
         \leqslant2^{t-1}\mathbb E[2^{-T}]
         =2^{t-1}e^{-t}<\frac12.
\end{align*}
The successful output for graph $x$ is the subnormalized
state $\gamma(\ket{\psi_x}\bra{\psi_x})^{\otimes t}$.
For $\omega_b:=\mathbb E_{x\sim P_b}\ketbra{\varphi_x}{\varphi_x}$,
the resulting flagged channel $\mathcal N$ therefore satisfies
\begin{align*}
    \mathcal N(\omega_b)
        &=\gamma\ketbra{\mathrm{ok}}{\mathrm{ok}}\otimes\rho_b^{(t)}
          +(1-\gamma)\ketbra{\mathrm{fail}}{\mathrm{fail}}\otimes\tau,
\end{align*}
where $\tau$ is fixed. We may include the parity compression
in the successful branch without changing the trace distance.
The failure blocks cancel, so contractivity gives
\begin{align}
    \gamma D^{(t)}
        &\leqslant\tfrac12\|\omega_1-\omega_0\|_1,
    & D^{(t)}&\leqslant\|\omega_1-\omega_0\|_1.
    \label{eq:phase-poisson-comparison}
\end{align}

\emph{Evaluating the weighted sum.}
In the product-state basis, let
$w(z)=p^{|z|}(1-p)^{M-|z|}$ for $z\in\{0,1\}^M$.
The amplitude at $z$ is
$\sqrt{w(z)}\prod_{e\in\operatorname{supp}(z)}s_e(x)$.
Using the correlations from Step~1 gives
\begin{align*}
    \omega_0&=\operatorname{diag}(w), &
    (\omega_1-\omega_0)_{zz'}
        &=\begin{cases}
            \sqrt{w(z)w(z')}\,q(z\oplus z'),&z\ne z',\\
            0,&z=z'.
          \end{cases}
\end{align*}
Lemma~\ref{lem:phase-weighted-trace-norm}, with $W=\omega_0$,
then yields
\begin{align*}
    \|\omega_1-\omega_0\|_1^2
        &\leqslant\sum_{z\ne z'}\sqrt{w(z)w(z')}\,q(z\oplus z')^2\\
        &=\sum_{d\ne\emptyset}q(d)^2
          \sum_z\sqrt{w(z)w(z\oplus d)}
         =\sum_{d\ne\emptyset}\beta^{|d|}q(d)^2.
\end{align*}
For the last equality, the inner sum factors over coordinates.
A coordinate outside $d$ contributes $1$; a coordinate in $d$
contributes $2\sqrt{p(1-p)}=\beta$.
Unlike the fixed-copy parity basis, this auxiliary basis is the
full cube, so the sum includes both odd and even $|d|$.
Combining this estimate with
Eq.~\eqref{eq:phase-poisson-comparison} proves the claim.
\end{proof}

\paragraph{Step 3: Bounding clique overlaps.}
\begin{proof}[Completion of the proof of Proposition~\ref{thm:phase-hidden-copy}]
Let $C,C'$ be independent uniform $k$-vertex sets and put
$R=|C\cap C'|$. Independence gives
$q(d)^2=\Pr[d\subseteq\binom{C\cap C'}2]$. Therefore
\begin{align}
    \bigl(D^{(t)}\bigr)^2
        &\leqslant\sum_{d\ne\emptyset}\beta^{|d|}q(d)^2\notag\\
        &=\mathbb E_{C,C'}\!\left[
          \sum_{\emptyset\ne d\subseteq\binom{C\cap C'}2}
          \beta^{|d|}\right]
         =\mathbb E\!\left[(1+\beta)^{\binom R2}-1\right].
    \label{eq:phase-overlap-bound}
\end{align}
The last equality is the binomial theorem: among the
$\binom R2$ common edges, summing over all subsets gives
$(1+\beta)^{\binom R2}$, and removing the empty set subtracts one.
The two clique locations arise from squaring $q(d)$;
the input still consists of copies of one graph.

Set $a=k^2/n<1$. For any fixed $r$-subset $A\subseteq C$,
\begin{align*}
    \Pr(A\subseteq C')
        &=\frac{\binom{n-r}{k-r}}{\binom nk}
         =\prod_{j=0}^{r-1}\frac{k-j}{n-j}
         \leqslant\left(\frac kn\right)^r.
\end{align*}
Each factor is at most $k/n$ because $k\leqslant n$.
A union bound over the $\binom kr$ possible subsets gives
\begin{align*}
    \Pr(R=r)
        &\leqslant\Pr(R\geqslant r)
         \leqslant\binom kr\left(\frac kn\right)^r
         \leqslant\frac{a^r}{r!}.
\end{align*}
The copy budget in Eq.~\eqref{eq:phase-hidden-budget}, together
with $1-e^{-u}\leqslant u$ and $2m=k(k-1)$, implies
\begin{align*}
    \beta(k-1)
        &\leqslant(k-1)\sqrt{\frac{8t}{M}}
         \leqslant\sqrt{\frac{k-1}{k}}\ln\frac1a
         \leqslant\ln\frac1a.
\end{align*}
Thus, for $2\leqslant r\leqslant k$,
\begin{align*}
    (1+\beta)^{\binom r2}
        &\leqslant\exp\!\left(\frac r2\beta(r-1)\right)
         \leqslant a^{-r/2}.
\end{align*}
The integrand in Eq.~\eqref{eq:phase-overlap-bound} vanishes
for $R=0,1$. Combining these estimates gives
\begin{align*}
    \bigl(D^{(t)}\bigr)^2
        &\leqslant\sum_{r=2}^k\frac{a^{r/2}}{r!}
         \leqslant\frac a2\sum_{r=2}^{\infty}\frac1{(r-2)!}
         =\frac e2a,
\end{align*}
where $a^{r/2}\leqslant a$ and $r!\geqslant2(r-2)!$.
Taking square roots proves Eq.~\eqref{eq:phase-hidden-distance}.
If $k=o(\sqrt n)$, this bound tends to zero, so constant
advantage requires eventually exceeding the stated copy budget.
For $k=\lfloor n^{1/2-\varepsilon}\rfloor$,
$M/m=\Theta(n^{1+2\varepsilon})$ and
$\ln(n/k^2)=(2\varepsilon+o(1))\ln n$, giving the claimed scale.
\end{proof}

\subsection{Detection by Separate Measurements}
\label{app:phase-linear-upper}
We prove Proposition~\ref{prop:phase-linear-upper} by recovering
relative edge signs. We first describe the measurement and full
reconstruction, then show why detection needs only partial
reconstruction. All copies encode the same graph $x$.

\paragraph{Pair-parity measurement.}
Use the readout from Hidden Matching~\cite{BJK08,KR06}.
For even $M$, choose a uniformly random perfect matching of the
edge coordinates and measure one copy in the basis
\begin{align*}
    \left\{\frac{\ket e+\ket f}{\sqrt2},
           \frac{\ket e-\ket f}{\sqrt2}:
           \{e,f\}\text{ is a matched pair}\right\}.
\end{align*}
The two outcomes for a pair have probabilities
$(s_e+s_f)^2/(2M)$ and $(s_e-s_f)^2/(2M)$.
Exactly one is nonzero, revealing $s_es_f$ without error.
Each matched pair is observed with probability $2/M$;
averaging over the matching makes the pair uniform among all
$\binom M2$ pairs, independently of $x$.

For odd $M\geqslant3$, leave one uniformly chosen coordinate
unmatched, match the rest uniformly, and include the unmatched
basis vector in the measurement. Discard this singleton outcome,
which has probability $1/M$. Every retained pair is again uniform.
Fresh independent matchings give independent uniform pairs.
Since each copy produces a pair with probability at least $2/3$,
a Chernoff bound shows that $4L$ copies supply at least $L$ pairs
except with probability at most $e^{-L/2}$. For even $M$, every
copy produces a pair.

Form a constraint graph on the $M$ edge coordinates, with each
observed pair labeled by $s_es_f$. Multiplying labels along a
path cancels the intermediate signs and gives the relative sign
of its endpoints. A spanning tree therefore recovers all signs
in a connected component up to one common flip. The matching
measurements and sign propagation take polynomial time.

\begin{proposition}[Reconstruction up to complementation]
\label{prop:separate-reconstruction}
For any fixed graph $x$, separate measurements on $O(M\log M)$
phase copies recover $\{x,\bar x\}$ with probability $1-o(1)$.
For $0<\delta<1/2$, a budget of
$O(M(\log M+\log(1/\delta)))$ copies reduces failure to at most
$\delta$.
\end{proposition}
\begin{proof}
The case $M=1$ is immediate. Otherwise, it suffices to connect
the constraint graph. A uniform pair crosses a fixed set of
$a\leqslant M/2$ coordinates with probability
\begin{align*}
    \frac{2a(M-a)}{M(M-1)}&\geqslant\frac aM.
\end{align*}
Every disconnected graph has a component of size at most $M/2$.
After $L$ retained pairs, a union bound gives
\begin{align*}
    \Pr[\text{disconnected}]
        &\leqslant\sum_{a=1}^{\lfloor M/2\rfloor}
          \binom Ma e^{-La/M}
         \leqslant\sum_{a\geqslant1}(Me^{-L/M})^a.
\end{align*}
This is $o(1)$ for $L\geqslant cM\ln M$ with fixed $c>1$,
and at most $\delta/2$ for $L\geqslant M\ln(4M/\delta)$.
With this latter choice, the probability of obtaining fewer
than $L$ pairs from $4L$ copies is also at most $\delta/2$.
Connectedness gives all signs up to a common flip, hence the
complement pair $\{x,\bar x\}$.
\end{proof}

\paragraph{Computational access after reconstruction.}
\label{app:phase-access}
Proposition~\ref{prop:separate-reconstruction} recovers a classical
description of $\{x,\bar x\}$ from polynomially many phase copies,
with inverse-polynomial error and polynomial running time.
Conversely, either representative lets us prepare any polynomial
number of phase copies efficiently. Thus these forms of access
can simulate one another, provided the full reconstruction budget
and polynomial-time quantum processing are allowed. This does not
imply classical simulability.

\paragraph{Detection from partial reconstruction.}
Full reconstruction permits a search for a $k$-clique in either
graph of the pair. For detection, it suffices to recover signs
on a large fraction of the coordinates, reducing the copy budget
to $O(M)$.

\begin{proof}[Proof of Proposition~\ref{prop:phase-linear-upper}]
Fix $\epsilon>0$ and choose $0<\eta<1/3$ so that
\begin{align}
    b_\eta:=-\ln\frac{1+\eta}{2}
        &>\frac{2\ln2}{2+\epsilon}.
    \label{eq:phase-partial-recovery-fraction}
\end{align}
This is possible since $b_\eta\to\ln2$ as $\eta\downarrow0$.
Choose a constant $c_\eta$ with
$2\eta(1-\eta)c_\eta>\ln2$ and set $L=\lceil c_\eta M\rceil$.
Measure $4L$ copies and keep the first $L$ retained pairs.
If too few pairs are obtained, output the null; this failure
has probability at most $e^{-L/2}=o(1)$.

\emph{Recovering a large component.}
A cut with between $\eta M$ and $(1-\eta)M$ coordinates is
crossed by a uniform pair with probability at least
$2\eta(1-\eta)$. A union bound over at most $2^M$ cuts gives
\begin{align}
    \Pr[\text{some such cut has no observed crossing}]
        &\leqslant2^M e^{-2\eta(1-\eta)L}=o(1).
    \label{eq:phase-large-component-bound}
\end{align}
If no component has size greater than $(1-\eta)M$, a union of
components has size in $[\eta M,(1-\eta)M]$: either one
component lies in this interval, or accumulating components
smaller than $\eta M$ gives a total in $[\eta M,2\eta M)$.
Here $2\eta<1-\eta$. Such a union has no crossing edge.
Thus, with probability $1-o(1)$, a unique component $S$ has
$|S|>(1-\eta)M$. Output the null if no such component exists.

Recover the signs in $S$ up to their common flip. Accept
planting if some $k$-vertex set $C$ has all recovered signs
on $\binom C2\cap S$ equal. This condition is unchanged by
the common flip, and the planted set always passes.

\emph{Bounding false positives.}
Conditional on $|S|=s$, the component $S$ is a uniform
$s$-element coordinate subset, independent of $x$.
The observed coordinate pairs are independent of $x$, and
their joint law is invariant under every coordinate permutation.
For a fixed candidate $C$, write $F=\binom C2$.
Under the null, the true signs on $F\cap S$ are independent
and fair. Since a common flip does not change whether they
are all equal, the candidate's passing probability is at most
$2\mathbb E[2^{-|F\cap S|}\mid |S|=s]$.
The bound also covers $F\cap S=\emptyset$.

For a fixed $B\subseteq F$, uniformity gives
\begin{align*}
    \Pr(B\cap S=\emptyset\mid |S|=s)
        &=\frac{\binom{M-s}{|B|}}{\binom M{|B|}}
         \leqslant\left(1-\frac sM\right)^{|B|}.
\end{align*}
The inequality multiplies successive avoidance probabilities
when sampling without replacement. Expanding
$2^{-|F\cap S|}=2^{-m}\prod_{e\in F}(1+\mathbf1[e\notin S])$
therefore yields
\begin{align*}
    \mathbb E[2^{-|F\cap S|}\mid |S|=s]
        &=2^{-m}\sum_{B\subseteq F}
          \Pr(B\cap S=\emptyset\mid |S|=s)\\
        &\leqslant\left(1-\frac{s}{2M}\right)^m
         \leqslant\left(\frac{1+\eta}{2}\right)^m.
\end{align*}
Averaging over the possible sizes and taking a union bound
across candidates gives
\begin{align*}
    \Pr_{P_0}[\text{accept}]
        &\leqslant2\binom nk\left(\frac{1+\eta}{2}\right)^m=o(1).
\end{align*}
Indeed, its logarithm is at most
$\ln2+k\ln n-b_\eta k(k-1)/2$, which tends to $-\infty$
by Eq.~\eqref{eq:phase-partial-recovery-fraction} and
$k\geqslant(2+\epsilon)\log_2n$.
The false-negative probability is at most the $o(1)$
probability of failing to obtain the pairs or the large component.
The total budget is $4L=O_\epsilon(M)$ copies.
The final search over vertex sets need not be efficient.
\end{proof}

\section{Symmetry-Adapted Measurements: Proofs}
\label{app:operator-reductions}
We first give the exact conditional Schur-label law, then
describe the readouts obtained by retaining multiplicity or
Specht. Finally, we identify the correlations preserved by
the label-and-Specht readout and use them to prove the
refined weak-label bound.

\subsection{Exact Conditional Label Probabilities}
\label{sec:conditional-label-laws}
The weight-$w$ space
$\mathcal K_w=\operatorname{span}\{\ket{x}:|x|=w\}$ is the
permutation representation on $w$-element edge subsets.
Young's rule gives~\cite{Sagan01}
\begin{align*}
    \mathcal K_w
        &\cong\bigoplus_{i=0}^{\min(w,M-w)}S^{(M-i,i)},
    & d_i&=\binom Mi-\binom M{i-1},
\end{align*}
where $\binom M{-1}=0$. The dimension formula follows by
subtracting the dimensions at successive weights.

\begin{lemma}[Isotypic sampling of a graph]
\label{prop:graph-probe-label}
For a graph $x$ with $w$ edges, the conditional label law is
\begin{align*}
    p(i\mid w)
        &=\bra{x}\Pi_i\ket{x}
         =\begin{cases}
             d_i/\binom Mw,&i\leqslant\min(w,M-w),\\
             0,&\text{otherwise}.
           \end{cases}
\end{align*}
\end{lemma}
\begin{proof}
Edge permutations act transitively on the basis of
$\mathcal K_w$ and commute with $\Pi_i$, so its diagonal
entries are equal there. Their sum is $d_i$ if the
representation occurs and zero otherwise. Divide by
$\dim\mathcal K_w=\binom Mw$.
\end{proof}
Consequently, for the edge count $W_b$ under hypothesis $b$,
\begin{align*}
    p_b(i)
        &=d_i\sum_{w=i}^{M-i}\frac{\Pr(W_b=w)}{\binom Mw},
    & p_0(i)&=\frac{r_i d_i}{N}.
\end{align*}
Here $W_0\sim\operatorname{Bin}(M,1/2)$ and
$W_1\overset{d}=m+\operatorname{Bin}(M-m,1/2)$.
These formulas also make $p(i\mid w)=p(i\mid M-w)$ explicit.

\subsection{Collective Averaging and Reduced Schur States}
\label{app:reduced-schur-states}

We first define the states available after measuring the
label and retaining either multiplicity or Specht. Write
the Schur output and its unnormalized diagonal label blocks as
\begin{align*}
    \widetilde\sigma_b
        &:=U_{\mathrm{Sch}}\sigma_bU_{\mathrm{Sch}}^\dagger,
    &
    \widetilde\sigma_{b,i}
        &:=\bra{i}_{\mathsf L}\widetilde\sigma_b\ket{i}_{\mathsf L}.
\end{align*}
Each block has trace $p_b(i)$. Tracing out multiplicity gives
\begin{align*}
    A_{b,i}
        &:=\operatorname{Tr}_{\mathsf M}
           (\widetilde\sigma_{b,i}),
    &
    \operatorname{Tr}A_{b,i}&=p_b(i).
\end{align*}
These are unnormalized Specht states: each includes the
probability of observing its label. The Specht readouts
with and without the label are therefore
\begin{align}
    \tau_b^{\mathsf{LS}}
        &=\sum_i\ketbra{i}{i}_{\mathsf L}\otimes A_{b,i},
    &
    \sigma_b^{\mathsf S}
        &=\sum_i A_{b,i}.
    \label{eq:specht-only-exact-states}
\end{align}
For the second sum, each block is extended by zero to the
common physical Specht register using the fixed encoding.

\paragraph{Measured label and multiplicity.}
If Specht is discarded instead, the output is
\begin{align*}
    \tau_b^{\mathsf{LM}}
        &:=\sum_i\ketbra{i}{i}_{\mathsf L}\otimes
          \operatorname{Tr}_{\mathsf S}(\widetilde\sigma_{b,i}).
\end{align*}

\begin{proposition}[Label and multiplicity are equivalent to edge count]
\label{prop:label-multiplicity-count}
For computational-basis graph inputs, the measured label
and multiplicity state can be prepared from the edge count,
and the edge count can be recovered exactly from that state.
Consequently,
\begin{align*}
    \tfrac12\|\tau_1^{\mathsf{LM}}-\tau_0^{\mathsf{LM}}\|_1
        &=\operatorname{TV}(W_1,W_0).
\end{align*}
\end{proposition}

\begin{proof}
Within each measured block, edge permutations act only on
the Specht factor, which is discarded. Thus the retained
state is unchanged by edge permutations and depends on a
graph only through its edge count. For suitable states
$\eta_w$,
\begin{align*}
    \tau_b^{\mathsf{LM}}
        &=\sum_w\Pr(W_b=w)\eta_w.
\end{align*}
Preparing $\eta_w$ from $w$ gives the upper bound by
data processing.

Conversely, the edge-count observable
$\widehat W:=\sum_{e\in\mathcal E}(\ket1\!\bra1)_e$
commutes with all edge permutations. In physical register
order $\mathsf L,\mathsf M,\mathsf S$, it therefore has the form
\begin{align*}
    U_{\mathrm{Sch}}\widehat WU_{\mathrm{Sch}}^\dagger
        &=\sum_i\ketbra{i}{i}_{\mathsf L}\otimes
          Q_i\otimes I_{d_i},
\end{align*}
where $Q_i$ is Hermitian on multiplicity. Since
$\widehat W\ket{x}=|x|\ket{x}$, measuring $Q_i$ conditional
on the recorded label recovers the input edge count exactly.
Data processing in this direction gives the reverse inequality.
\end{proof}

\paragraph{Collective averaging and the Specht readouts.}
To calculate the blocks $A_{b,i}$, we simplify the input
by averaging over collective rotations:
\begin{align}
    \mathcal C(A)
        &:=\int_{\mathrm{SU}(2)}
           U^{\otimes M}A(U^\dagger)^{\otimes M}\,dU.
    \label{eq:collective-rotation-channel}
\end{align}
Here $dU$ is normalized Haar measure, and the same rotation
acts on every edge qubit. Collective rotations act on the
multiplicity spaces through mutually inequivalent irreducible
representations. Schur's lemma therefore gives, in register
order $\mathsf L,\mathsf M,\mathsf S$,
\begin{align*}
    U_{\mathrm{Sch}}\mathcal C(\sigma_b)U_{\mathrm{Sch}}^\dagger
        &=\sum_i\ketbra{i}{i}_{\mathsf L}\otimes
          \frac{I_{r_i}}{r_i}\otimes A_{b,i}.
\end{align*}
Thus averaging removes coherences between labels and
replaces multiplicity by a maximally mixed state, while
preserving every $A_{b,i}$~\cite{BRS07}.  Discarding multiplicity from this averaged state gives
$\tau_b^{\mathsf{LS}}$. Conversely, appending
$I_{r_i}/r_i$ conditional on the label reconstructs it.
The averaged state therefore contains exactly the
information in the label-and-Specht readout. It also
has the same Specht-only marginal:
\begin{align}
    \operatorname{Tr}_{\mathsf L,\mathsf M}
       \!\left(
          U_{\mathrm{Sch}}\mathcal C(\sigma_b)
          U_{\mathrm{Sch}}^\dagger
       \right)
        &=\sigma_b^{\mathsf S}.
    \label{eq:collective-average-preserves-specht}
\end{align}
We may therefore calculate both readouts from the
collectively averaged states. 

For a candidate clique $C$, let $\mathcal E_C=\binom C2$.
Define the projector that symmetrizes its $m$ edge qubits,
and average these projectors over all candidate locations:
\begin{align}
    \Gamma_C
        &:=\frac1{m!}\sum_{\pi\in S_{\mathcal E_C}}R(\pi),
    &
    \Omega_k
        &:=\frac1{\binom nk}\sum_{|C|=k}\Gamma_C.
    \label{eq:specht-candidate-projectors}
\end{align}
Here $S_{\mathcal E_C}$ permutes the clique-edge positions
and fixes all others. The symmetric subspace of $m$ qubits
has dimension $m+1$, so
\begin{align*}
    \operatorname{rank}\Gamma_C
        &=(m+1)2^{M-m}=N\tau_m,
    &
    \tau_m&:=(m+1)2^{-m}.
\end{align*}
Consequently, $\Omega_k$ is a positive contraction with
$\operatorname{Tr}\Omega_k=N\tau_m$.

\begin{proposition}[The planted state after collective averaging]
\label{prop:specht-projector-mixture}
The collectively averaged graph states satisfy
\begin{align}
    \mathcal C(\sigma_0)&=\frac IN,
    &
    \mathcal C(\sigma_1)&=\frac{\Omega_k}{N\tau_m}.
    \label{eq:specht-projector-mixture}
\end{align}
\end{proposition}

\begin{proof}
The maximally mixed null is unchanged by unitary conjugation.
Conditional on planting at $C$, the input is
\begin{align*}
    \sigma_C
        &=(\ket1\!\bra1)^{\otimes m}_{\mathcal E_C}
          \otimes
          \frac{I_{\mathcal E\setminus\mathcal E_C}}{2^{M-m}}.
\end{align*}
Collective rotations leave the second factor unchanged.
On the clique edges, every rotated vector
$(U\ket1)^{\otimes m}$ lies in the symmetric subspace.
The averaged state is invariant under collective rotations,
which act irreducibly on this subspace. Schur's lemma and
trace normalization therefore give
\begin{align*}
    \int_{\mathrm{SU}(2)}
       (U\ket1\!\bra1 U^\dagger)^{\otimes m}\,dU
        &=\frac{\Pi_{\mathrm{sym}}^{(m)}}{m+1};
\end{align*}
see also~\cite[Proposition~6]{Harrow13}.
Hence
\begin{align*}
    \mathcal C(\sigma_C)
        &=\frac{\Gamma_C}{(m+1)2^{M-m}}
         =\frac{\Gamma_C}{N\tau_m}.
\end{align*}
Averaging over the uniformly chosen location $C$ replaces
$\Gamma_C$ by $\Omega_k$, proving the claim.
\end{proof}

The proposition describes the label-and-Specht experiment
through the statistically equivalent averaged graph states.
The Specht-only planted state follows by tracing out the
label and multiplicity:
\begin{align}
    \sigma_1^{\mathsf S}
        &=\frac1{N\tau_m}
          \operatorname{Tr}_{\mathsf L,\mathsf M}
          \!\left(
             U_{\mathrm{Sch}}\Omega_k
             U_{\mathrm{Sch}}^\dagger
          \right).
    \label{eq:specht-planted-from-omega}
\end{align}

For the averaged graph states,
\begin{align*}
    \mathcal C(\sigma_1)-\mathcal C(\sigma_0)
        &=\frac{\Omega_k-\tau_m I}{N\tau_m}.
\end{align*}
Their optimal equal-prior accepting effect is therefore
\begin{align*}
    E_\star
        &=\mathbf1_{(\tau_m,\infty)}(\Omega_k).
\end{align*}
This
also specifies an optimal test for the label-and-Specht
readout. The formula gives a measurement target but no
efficient implementation; an implementable effect with
constant distinguishing gap would suffice.

%
%

\subsection{Operator symmetries}
\label{sec:operator-symmetry}

Probability vectors and their diagonal density operators
have the same permutation symmetries, but these symmetries
act differently on the two objects. For a function $f$ on
graphs, let $D_f:=\sum_x f(x)\ketbra{x}{x}$. Then
\begin{align*}
    D_{R(\pi)f}&=R(\pi)D_fR(\pi)^\dagger.
\end{align*}
Thus $R(\pi)P_b=P_b$ holds exactly when
$R(\pi)\sigma_bR(\pi)^\dagger=\sigma_b$.
The former is invariance of a vector; the latter is
invariance of an operator under conjugation. In particular,
$P_0$ lies in the trivial vector sector, whereas
$\sigma_0=I/N$ has support in every Schur sector.
Each isotypic projector is itself invariant under
conjugation, including those for nontrivial labels.
Commutation with the action therefore does not force
a density operator to be supported on invariant vectors.

We identify the edge correlations retained by the measured
label and Specht register, and those lost when only the label
is kept. Recall that the collective rotation average
$\mathcal C$ from Eq.~\eqref{eq:collective-rotation-channel}
contains exactly the information in the label-and-Specht
readout. Define the additional edge-permutation average
\begin{align*}
    \mathcal T(A)
        &:=\frac1{M!}\sum_{\pi\in S_M}
           R(\pi)AR(\pi)^\dagger.
\end{align*}
Collective averaging makes multiplicity maximally mixed
within each block; edge averaging then does the same to
Specht. Their composition therefore retains only the label
probabilities:
\begin{align}
    \mathcal T\mathcal C(A)
        &=\sum_i\frac{\operatorname{Tr}(\Pi_iA)}{d_i r_i}\Pi_i.
    \label{eq:operator-label-channel}
\end{align}
Both averages are orthogonal projections for the
Hilbert--Schmidt inner product
$\langle A,B\rangle_{\mathrm{HS}}:=\operatorname{Tr}(A^\dagger B)$
and commute with each other~\cite{BRS07}.

For diagonal graph states, $\mathcal T$ averages probabilities
within each edge-count class. As in
Eq.~\eqref{eq:symmetry-signal-decomposition}, decompose
$\Delta:=\sigma_1-\sigma_0$ into
\begin{align*}
    \Delta
        &=\Delta_{\mathrm{count}}+\Delta_{\mathrm{str}},
    &
    \Delta_{\mathrm{count}}&:=\mathcal T(\Delta).
\end{align*}
The structural part $\Delta_{\mathrm{str}}$ has zero average
within every edge-count class. Because
$\mathcal T(\Pi_i)=\Pi_i$ and $\mathcal T$ is self-adjoint,
$\operatorname{Tr}(\Pi_i\Delta_{\mathrm{str}})=0$:
weak Schur sampling cannot detect this part of the signal.
We now determine what collective averaging preserves.

\subsubsection{Correlations retained after discarding multiplicity}

For an edge subset $d$, let $Z_d:=\prod_{e\in d}Z_e$, where
$Z_e$ is Pauli $Z$ on edge qubit $e$ and the identity elsewhere.
The diagonal signal has the Fourier expansion
\begin{align*}
    \Delta&=\frac1N\sum_{d\ne\emptyset}a_dZ_d,
    &
    a_d&:=\mathbb E_{X\sim P_1}(-1)^{\sum_{e\in d}X_e}.
\end{align*}
Every nonempty sign product has null mean zero. Under planting,
it has nonzero mean only when all edges in $d$ lie inside
the planted set, in which case every selected bit is one.
Thus, using the notation from the phase-state proof,
\begin{align*}
    a_d&=(-1)^{|d|}q(d),
    &
    q(d)&=\frac{\binom{n-r(d)}{k-r(d)}}{\binom nk},
    \qquad r(d):=\left|\bigcup_{e\in d}e\right|.
\end{align*}
Here $r(d)$ counts touched vertices, and the numerator is zero
when $r(d)>k$. The factor $(-1)^{|d|}$ accounts for the sign
convention: Pauli $Z$ uses $(-1)^{X_e}$, whereas the phase-state
proof uses $s_e(X)=2X_e-1=-(-1)^{X_e}$.

\begin{proposition}[Exact survival of even Fourier degrees]
\label{prop:operator-even-survival}
For edge subsets $d,d'$ and $r=|d|$,
\begin{align*}
    \frac1N
       \bigl\langle\mathcal C(Z_d),\mathcal C(Z_{d'})\bigr\rangle_{\mathrm{HS}}
        &=\mathbf1[d=d']
          \begin{cases}
              (r+1)^{-1},&r\text{ even},\\
              0,&r\text{ odd}.
          \end{cases}
\end{align*}
Consequently, $\mathcal C$ annihilates every odd-degree
diagonal Fourier component. On the span of the degree-$r$ strings for even
$r$, it is injective and multiplies squared Hilbert--Schmidt
norms by $(r+1)^{-1}$.
\end{proposition}

\begin{proof}
For Haar-random $U$, the operator $UZU^\dagger$ is
$\mathbf n\cdot\boldsymbol\sigma$, where $\mathbf n$ is
uniform on the unit sphere. Conjugation preserves the set
of nonidentity positions of each Pauli string, so distinct
sets remain orthogonal. Since $\mathcal C$ is an orthogonal
projection,
\begin{align*}
    \frac1N
       \bigl\langle\mathcal C(Z_d),\mathcal C(Z_{d'})\bigr\rangle_{\mathrm{HS}}
        &=\frac1N\langle Z_d,\mathcal C(Z_{d'})\rangle_{\mathrm{HS}}\\
        &=\mathbf1[d=d']\,\frac12\int_{-1}^1u^r\,du.
\end{align*}
The last equality uses that $n_z$ is uniform on $[-1,1]$.
This average is zero for odd $r$ and $(r+1)^{-1}$ for even
$r$. The norm and injectivity statements follow because
the strings $Z_d$ are orthogonal with squared norm $N$.
\end{proof}

For example, two distinct edges satisfy
\begin{align*}
    \mathcal C(Z_eZ_f)
        &=\frac13(X_eX_f+Y_eY_f+Z_eZ_f),\qquad e\ne f,
\end{align*}
where $X_e,Y_e,Z_e$ are the Pauli operators on qubit $e$.
For $n\geqslant4$ and $3\leqslant k<n$, an adjacent pair
touches three vertices and has a different planted coefficient
from a disjoint pair, which touches four. Collective averaging
preserves this difference. Edge averaging assigns the same
coefficient to every pair and erases it. The label-and-Specht
state therefore retains arrangement information absent from
the weak label.

\subsubsection{Structural signal and label probabilities}

Edge averaging replaces every coefficient of a given degree
by their mean. The preceding proposition lets us quantify both
the variation retained by $\mathcal C$ and the signal left
after the additional average $\mathcal T$.
Write $\|A\|_2^2:=\langle A,A\rangle_{\mathrm{HS}}$.

\begin{corollary}[Structural signal and weak-label second moment]
\label{cor:operator-structural-norm}
For $1\leqslant r\leqslant M$, let
\begin{align*}
    \bar a_r
        &:=\frac1{\binom Mr}\sum_{|d|=r}a_d
         =(-1)^r\frac{\binom mr}{\binom Mr}.
\end{align*}
Then
\begin{align*}
    N\|\mathcal C(\Delta_{\mathrm{str}})\|_2^2
        &=\sum_{\substack{2\leqslant r\leqslant M\\r\ \mathrm{even}}}
          \frac1{r+1}\sum_{|d|=r}|a_d-\bar a_r|^2.
\end{align*}
The weak-label laws satisfy
\begin{align}
    \chi^2(p_1\|p_0)
        &:=\sum_i\frac{(p_1(i)-p_0(i))^2}{p_0(i)}\notag\\
        &=N\|\mathcal T\mathcal C(\Delta)\|_2^2
         =\sum_{\substack{2\leqslant r\leqslant m\\r\ \mathrm{even}}}
          \frac{\binom mr^2}{(r+1)\binom Mr}.
    \label{eq:weak-schur-label-chi-square}
\end{align}
\end{corollary}

\begin{proof}
Each planted location contains exactly $\binom mr$ edge
subsets of size $r$. Summing their correlations gives
$\sum_{|d|=r}a_d=(-1)^r\binom mr$, proving the expression
for $\bar a_r$.
The Fourier coefficients of $\Delta_{\mathrm{str}}$ and
$\mathcal T(\Delta)$ are respectively
$a_d-\bar a_{|d|}$ and $\bar a_{|d|}$.
Proposition~\ref{prop:operator-even-survival} gives the
first identity and, since $\mathcal C$ commutes with
$\mathcal T$,
\begin{align*}
    N\|\mathcal T\mathcal C(\Delta)\|_2^2
        &=\sum_{\substack{2\leqslant r\leqslant M\\r\ \mathrm{even}}}
          \frac{\binom Mr\bar a_r^2}{r+1}.
\end{align*}
Substituting $\bar a_r$ gives the sum in
Eq.~\eqref{eq:weak-schur-label-chi-square}.
Finally, Eq.~\eqref{eq:operator-label-channel} makes
$\mathcal T\mathcal C(\Delta)$ scalar on block $i$, with
value $(p_1(i)-p_0(i))/(d_i r_i)$. Since
$p_0(i)=d_i r_i/N$,
\begin{align*}
    N\|\mathcal T\mathcal C(\Delta)\|_2^2
        &=N\sum_i\frac{(p_1(i)-p_0(i))^2}{d_i r_i}
         =\chi^2(p_1\|p_0).\qedhere
\end{align*}
\end{proof}

Thus collective averaging preserves every even-degree
coefficient, while the weak label retains only their mean
at each degree. Appendix~\ref{app:weak-schur-refinement}
uses the exact second moment to bound the weak-label distance.

These calculations describe the label-and-Specht readout.
A Specht-only measurement must use the same effect $F$ for
every label. On the original graph register, its effect is
\begin{align}
    E_F
        &:=U_{\mathrm{Sch}}^\dagger
          (I_{\mathsf L}\otimes I_{\mathsf M}\otimes F)
          U_{\mathrm{Sch}},
    &
    0&\leqslant F\leqslant I_{\mathsf S}.
    \label{eq:specht-only-effect}
\end{align}
This formula uses the fixed encoding into the common Specht
register. Since $\mathcal C(E_F)=E_F$, its acceptance gap
depends only on $\mathcal C(\Delta)$.
Corollary~\ref{thm:retain-specht}, proved using the rank bound,
guarantees such an effect with near-perfect discrimination
above the fixed-margin logarithmic threshold. The correlation
formulas identify the available signal but give no efficient
implementation of this effect.

\subsection{A Refined Bound for Weak Schur Sampling}
\label{app:weak-schur-refinement}
Complementary graphs have the same label distribution:
$X^{\otimes M}$ commutes with each isotypic projector and
maps $\ket{x}$ to $\ket{\bar x}$.
The absence of odd-degree terms in the exact second moment
from Appendix~\ref{sec:operator-symmetry} gives the sharper rate.

\begin{proposition}[Refinement using graph complementation]
\label{prop:weak-schur-refinement}
For $n\geqslant3$ and $2\leqslant k\leqslant n$, set
$\theta:=\frac{m^2}{M-m+1}$. Then
\begin{align}
    D_{\mathsf L}
        &\leqslant\min\left\{1,\frac12\sqrt{\cosh\theta-1}\right\}.
    \label{eq:weak-schur-refined-bound}
\end{align}
In particular, $D_{\mathsf L}=O(k^4/n^2)=o(1)$ when
$k=o(\sqrt n)$.
\end{proposition}
\begin{proof}
By Eq.~\eqref{eq:weak-schur-label-chi-square},
\begin{align*}
    \chi^2(p_1\|p_0)
        &=\sum_{\substack{2\leqslant r\leqslant m\\r\ \mathrm{even}}}
          \frac{\binom mr^2}{(r+1)\binom Mr}.
\end{align*}
For $r\leqslant m$,
\begin{align*}
    \frac{\binom mr^2}{\binom Mr}
        &=\frac1{r!}\prod_{j=0}^{r-1}\frac{(m-j)^2}{M-j}
         \leqslant\frac{\theta^r}{r!}.
\end{align*}
Dropping $(r+1)^{-1}$ and summing over even $r\geqslant2$
gives $\chi^2(p_1\|p_0)\leqslant\cosh\theta-1$.
Cauchy--Schwarz then yields
\begin{align*}
    D_{\mathsf L}
        &=\frac12\sum_i p_0(i)
          \left|\frac{p_1(i)}{p_0(i)}-1\right|
         \leqslant\frac12\sqrt{\chi^2(p_1\|p_0)}
         \leqslant\frac12\sqrt{\cosh\theta-1}.
\end{align*}
Together with $D_{\mathsf L}\leqslant1$, this proves the bound.
If $k=o(\sqrt n)$, then $\theta=O(k^4/n^2)=o(1)$ and
$\cosh\theta-1=O(\theta^2)$, giving the claimed rate.
\end{proof}
For $k=\lfloor n^{1/2-\varepsilon}\rfloor$ with fixed
$0<\varepsilon<1/2$, the optimal equal-prior success probability
using only the label is therefore
$\tfrac12+\tfrac12D_{\mathsf L}=\tfrac12+O(n^{-4\varepsilon})$.

\bibliographystyle{amsalpha}
\bibliography{biblio-master}
\end{document}